\documentclass[11pt,final]{article}
\usepackage{typearea}  % set page margins

\typearea{14}
\usepackage[colorlinks=true,citecolor=blue,urlcolor=blue,bookmarks=false]{hyperref}
\usepackage[numbers,sort]{natbib}
\usepackage{amsfonts} % e.g. mathbb
\usepackage{amsmath,amsthm,amssymb}

\usepackage{thmtools}
\usepackage{thm-restate}
\usepackage{cleveref}
\usepackage{xcolor}
\usepackage{xspace}
\usepackage{bm}
\usepackage{todonotes}
\usepackage{enumitem}

\theoremstyle{plain}
\newtheorem{theorem}{Theorem}[section]
\newtheorem*{theorem*}{Theorem}
\newtheorem{corollary}[theorem]{Corollary}

\newtheorem{lemma}[theorem]{Lemma}

\newtheorem{definition}[theorem]{Definition}

\newcommand{\E}{\mathbb{E}}

\newcommand{\calP}{\mathcal{P}}
\newcommand{\poly}{\mathrm{poly}}
\newcommand{\spec}{\mathcal{S}_{G}(\alpha)}

\newcommand{\ind}[1]{\mathbb{I}\left[#1\right]}
\newcommand{\prob}[1]{\mathbb{P}\left[#1\right]}
\newcommand{\expec}[1]{\mathbb{E}\left[#1\right]}
\newcommand{\bigO}[1]{\mathcal{O}\!\left(#1\right)}

\newcommand{\hop}{\text{hop}}

\newcommand{\obl}{\text{cong}}
\newcommand{\oblreal}{\text{cong}_{\mathbb{R}}}
\newcommand{\oblint}{\text{cong}_{\mathbb{Z}}}
\newcommand{\siz}{\text{siz}}
\newcommand{\cut}{\text{cut}_{G}}
\newcommand{\cou}{\text{cnt}_{G}}
\newcommand{\opt}{\text{opt}_{G, \mathbb{R}}}
\newcommand{\optint}{\text{opt}_{G, \mathbb{Z}}}
\newcommand{\supp}{\text{supp}}
\newcommand{\dil}{\text{dil}}
\newcommand{\zeroonedemand}{$\{0, 1\}$-demand\xspace}
\newcommand{\zeroonedemands}{$\{0, 1\}$-demands\xspace}

\title{Sparse Semi-Oblivious Routing:\\{\large Few Random Paths Suffice}\footnote{Supported in part by NSF grants CCF-1814603, CCF-1910588, NSF CAREER award CCF-1750808, a Sloan Research Fellowship, and funding from the European Research Council (ERC) under the European Union's Horizon 2020 research and innovation program (grant agreement No. 949272).}}  % ARXIV title!!!
\author{
   Goran Zuzic\\
   \small ETH Zurich\\
   \small goran.zuzic@inf.ethz.ch
   \and\textcircled{r}\footnote{The author ordering was randomized using \url{https://www.aeaweb.org/journals/policies/random-author-order/generator}. We kindly ask that citations of this work list the authors separated by \texttt{\textbackslash textcircled\{r\}} instead of commas: Zuzic \textcircled{r} Haeupler \textcircled{r} Roeyskoe.}\and
   Bernhard Haeupler\\  
   \small ETH Z\"urich and\\
   \small Carnegie Mellon University \\
   \small bernhard.haeupler@inf.ethz.ch
   \and\textcircled{r}\and
   Antti Roeyskoe\\
   \small ETH Z\"urich\\
   \small aroeyskoe@ethz.ch
 } 

\date{}

\begin{document}

%New draft
%\begin{abstract}
%The Packet routing problem in a network $G=(V,E)$ consists of a demand given by a set of $(s_i,t_i)$ node pairs, where

%% arxiv requires this .. comment out for PODC
\maketitle

\begin{abstract}
  The packet routing problem asks to select routing paths that minimize the maximum edge congestion for a set of packets specified by source-destination vertex pairs. We revisit a semi-oblivious approach to this problem: each source-destination pair is assigned a small set of well-chosen predefined paths before the demand is revealed, while the sending rates along the paths can be optimally adapted to the demand. This approach has been considered in practice in network traffic engineering due to its superior robustness and performance as compared to both oblivious routing and traditional traffic engineering approaches.
  
  \medskip

  We show the existence of sparse semi-oblivious routings: only $\mathcal{O}(\log n)$ paths are selected between each pair of vertices. The routing is $(\poly \log n)$-competitive \emph{for all demands} against the offline-optimal congestion objective, even on worst-case graphs. Even for the well-studied case of hypercubes, no such result was known: our deterministic and oblivious selection of $\mathcal{O}(\log n)$ paths is the first simple construction of a deterministic oblivious structure that near-optimally assigns source-destination pairs to few routes. Prior work shows that a deterministic selection of a single path in a hypercube yields unacceptable performance; our results contrast the current solely-negative landscape of results for semi-oblivious routing. We give the sparsity-competitiveness trade-off for lower sparsities and nearly match it with a lower bound.

  \medskip

  Our construction is extremely simple: Sample the few paths from any competitive oblivious routing. Indeed, this natural construction was used in traffic engineering as an unproven heuristic. We give a satisfactory theoretical justification for their empirical effectiveness: the competitiveness of the construction improves exponentially with the number of paths. In other words, semi-oblivious routing benefits from the power of random choices. Finally, when combined with the recent hop-constrained oblivious routing, we also obtain sparse and competitive structures for the completion-time objective.
\end{abstract}

%%%%%% PODC requires this after .. comment out for arxiv  %%%%%%%%%%%%%%%%%%%%%%
%\maketitle
%%%%%%%%%%%%%%%%%%%%%%%%%%%%%%%%%%%%%%%%%%%%%%%%%%%%%%%%%%%%%%%%%%%%%%%%%%%%%%%%

%%%%% ARXIV stuff .. comment out for PODC %%%%%%%%%%%%%%%%%%%%%%%%%%%%%%%
\thispagestyle{empty}
\setcounter{page}{0}

\newpage

\tableofcontents
\thispagestyle{empty} % in case you don't want to number this page
\setcounter{page}{0} % reset the page counter

\newpage
%%%%%%%%%%%%%%%%%%%%%%%%%%%%%%%%%%%%%%%%%%%%%%%%%%%%%%%%%%%%%%%%%%%%%%%%%

\section{Introduction}

Packet routing through a communication network is a fundamental task that is well-studied in both theoretical and practical contexts. We consider the following version of the task.

\begin{center}
\begin{minipage}{0.85\textwidth}
  \emph{The network is abstracted as an $n$-vertex undirected graph $G = (V, E)$. Initially, the network receives several packet delivery requests, where the $i^{th}$ packet should be transmitted from source $s_i \in V$ to destination $t_i \in V$. The goal is to select a path for each packet in a way that minimizes the maximum edge congestion, i.e., minimizes the maximum number of packets passing over any one edge.}
\end{minipage}
\end{center}

An offline version of this task is known as the multicommodity flow (MCF) problem, with different packets representing different commodities. However, solving MCF typically assumes knowledge of the entire set of packets upfront, a requirement that is often very restrictive. For this reason, a particularly appealing routing strategy is the so-called \emph{oblivious routing}, where each packet is routed independently (i.e., obliviously) from other packets using a predefined policy while at the same time requiring that the lump of all traffic near-optimally utilizes the network. The line of work on oblivious routings culminated with the celebrated result of Raecke~\cite{Racke08}, which proves that in every graph one can obliviously route the packets while guaranteeing the maximum congestion (of the most used edge) is $\bigO{\log n}$-competitive with the globally-and-offline-optimal maximum congestion. %The mere existence of such graph-theoretic combinatorial structures have led to numerous advancements in approximation algorithms on graphs~\cite{Racke08,gupta2006oblivious}.

Motivated by this success, a prominent extension called \emph{semi-oblivious routing} was suggested by Hajiaghayi, Kleinberg, and Leighton~\cite{HajiaghayiKL07} for VLSI routing and network traffic engineering in the hope of surpassing the competitiveness of (significantly more stringent) oblivious routing. A semi-obliviously-routed packet is required to obliviously specify a small set of \emph{candidate paths} it can possibly traverse. However, the final choice of over which path the packet is routed is made in a globally-optimal manner \emph{after} the entire set of packets (i.e., the so-called demand matrix) is revealed.

Unfortunately, the theoretical inquiries into the effectiveness of semi-oblivious routing have only yielded negative results: \cite{HajiaghayiKL07} showed that any such routing with polynomially-many candidate paths cannot be $o(\frac{\log n}{\log \log n})$-competitive in terms of edge congestion --- essentially no better than standard oblivious routing.

In contrast to the theoretical barriers, semi-oblivious solutions to packet routing have found notable success in traffic engineering since installing a new candidate path takes considerable effort that involves updating forwarding tables on geographically-distributed switches; on the other hand, the sending rates over candidate paths can be updated quickly (e.g., using a small snapshot of the global traffic every 15 seconds~\cite{KumarYYFKLLS18}). Semi-oblivious routing solutions offer superior performance as compared to traditional traffic-engineering approaches and they offer robustness over standard oblivious routing as the set of candidate paths can be chosen more diversely~\cite{KumarYYFKLLS18,kumar18semi}.

\subsection{Our Results and Consequences}

We provide the first theoretical evidence that semi-oblivious routing is significantly stronger than oblivious routing. As a simple subcase of our results, we show that each graph has a $\bigO{\log n}$-sparse semi-oblivious routing that is $(\poly \log n)$-competitive with the offline optimum on all permutation demands. Here, $\alpha$-sparse means that $\alpha$ paths are chosen between each pair of vertices (for a total of $\alpha n^2$ paths). We note that no standard oblivious routing can be $(\poly \log n)$-sparse, at least without a near-linear competitive ratio. This has compelling consequences.

\textbf{Consequence: Deterministic Routing.} One cannot in general \emph{deterministically} assign source-destination pairs to a single path without compromising on either obliviousness (dependencies between pairs) or competitiveness. Indeed, even on the widely-studied case of hypercubes, %(e.g., \cite{harary1988survey})
it is known that the best deterministic oblivious routing has competitiveness $\tilde{\Theta}(\sqrt{n})$~\cite{KaklamanisKT91}\footnote{Tilde (e.g., $\tilde{\Theta}$ or $\tilde{\mathcal{O}}$) hides $\poly \log n$ factors.}. We contribute one way to bypass this barrier: \emph{deterministically select a few paths} instead of one. Other methods of deterministic routing of hypercubes were developed, but they mostly involved complicated sorting networks~\cite{BorodinH85,ajtai19830} or related ideas~\cite{kuszmaul1990fast,grammatikakis1998packet}. Our deterministic strategy of selecting paths is far simpler, and moreover, our approach works for any graph. It is the first deterministic and oblivious strategy for general graphs that is $(\poly \log n)$-competitive.

\textbf{Consequence: Power of Random Choices in Semi-Oblivious Routing.} Consider defining $\alpha$-sparse classic oblivious routings in an analogous way, where the support size of paths between each pair of vertices is at most $\alpha$. Then, $1$-sparse oblivious routing corresponds to the deterministic case and the $\tilde{\Omega}(\sqrt{n})$ barrier~\cite{KaklamanisKT91} applies. As a simple corollary, any $\alpha$-sparse oblivious routing can at best be $\tilde{\Omega}(\sqrt{n} / \alpha)$-competitive. On the other hand, our results show that the competitiveness of an $\alpha$-sparse semi-oblivious routing is $\tilde{\mathcal{O}}(n^{\mathcal{O}(1/\alpha)})$; the competitiveness improves exponentially with $\alpha$. In other words, each additional path leads to a polynomial improvement in competitiveness; semi-oblivious routing benefits from the ``power of a few random choices'', an analog to the fundamental classic result where providing two random choices drastically improves performance~\cite{richa2001power}. This gives the first compelling theoretical separation between oblivious and semi-oblivious routings.

\textbf{Consequence: A Natural Construction and Its Traffic Engineering Applications.} Our $\alpha$-sparse semi-oblivious construction is extremely simple: \emph{for each pair of vertices, sample $\alpha$ random paths from any good oblivious routing distribution}. Indeed, this natural approach was considered for network traffic engineering~\cite{KumarYYFKLLS18,kumar18semi}: they sample from an oblivious routing distribution, due to domain constraints only sampling a small $\alpha$ number of paths between each pair of nodes (e.g., they choose $\alpha = 4$), and adapt the sending rates on these paths at real-time (semi-obliviousness). While they empirically find this approach to be very effective, the approach was an unproven heuristic with no a priori reason to work well. The observed ``power of a few random choices'' offers a compelling theoretical justification for why choosing a small constant sparsity like $\alpha = 4$ is a practical sweet spot that offers both adequate competitiveness and sparsity. Moreover, our paper not only explains why the approach works well for networks that occur in practice \cite{KumarYYFKLLS18}, but shows competitiveness for worst-case networks as well.

\textbf{Technical Challenges.} While the construction is conceptually extremely simple, its analysis in the context of semi-oblivious routings is deeply involved. It is easy to show using simple randomized rounding~\cite{raghavan1987randomized} that a random $\mathcal{O}(\log n)$-sparse routing is competitive on a fixed demand with high probability. The main challenge arises because the number of possible demands is exponential in $n$, hence the approach fails to be competitive \emph{on all demands}. Indeed, randomized rounding is entirely oblivious, while our results provably require us to exploit semi-obliviousness. A much more intricate analysis is required. In a nutshell, we prove the sampling strategy works using the probabilistic method~\cite{alon2016probabilistic}. We set up a randomized dynamic process: For a fixed demand, pretend to send packets on all candidate paths at once, and delete the edges that get \emph{overcongested} (together with all candidate paths crossing that edge). The goal is to show that, with exponential concentration, many candidate paths remain in the end, hence we can route along them without overcongestions. This enables the union bound over exponentially many demands. The main challenge in formalizing this argument is the dependence between path deletions, and resolving them requires a lot of care.

%The main challenge arises because the number of possible demands is exponential in $n$, while a random $\mathcal{O}(\log n)$-sparse solution can typically handle only polynomially-many demands. Indeed, simple randomized rounding~\cite{raghavan1987randomized} (i.e., a random $\mathcal{O}(\log n)$-sparse solution sampled from oblivious routings) without any semi-obliviousness will be competitive on a fixed polynomially-large set of demands. 

\textbf{Consequence: Optimizing Completion Time.} So far, we focused on minimizing the maximum edge congestion. However, the true objective in traffic engineering~\cite{KumarYYFKLLS18,kumar18semi} was to minimize the \emph{completion time}, i.e., the time until all packets arrive at their destinations (also known as the makespan, or minimizing the delay). The papers~\cite{KumarYYFKLLS18,kumar18semi} only explicitly optimize for congestion, but empirically conclude that this choice implicitly also adequately optimizes the completion time, at least on their benchmark networks.

Unfortunately, this inherently fails on worst-case instances: there exist graphs where optimizing congestion yields non-competitive completion-time guarantees~\cite{GhaffariHZ21}. We show that using the natural sampling strategy with the very recent hop-constrained oblivious routings~\cite{GhaffariHZ21}, we can obtain $\mathcal{O}(\log n)$-sparse semi-oblivious routings that are $(\poly\log n)$-competitive in terms of completion time, even on worst-case instances.

\textbf{Paper organization.} We first overview important concepts in \Cref{sec:our-concepts} followed by an overview of our results in \Cref{sec:formal-results}. Related work is presented in \Cref{sec:related-work} and the notation used throughout this paper in \Cref{sec:notation}. In \Cref{sec:sparse-semi-oblivious-routing}, we formally define and construct (fractional) semi-oblivious routings by sampling the candidate paths from an oblivious routing. The section is divided into a technical overview (\Cref{sec:technical-overview}), a proof structure overview (\Cref{sec:proof-structure}), the proof of the main lemma (\Cref{sec:proving-the-main-lemma}) and the reduction from the main Theorem to the Main Lemma (\Cref{sec:finishing-the-reduction}). The results for integral semi-oblivious routings are presented in \Cref{sec:integral-oblivious}, and the results for completion-time-competitive semi-oblivious routings in \Cref{sec:hopconstrainedsemiobl}. The lower bound that shows our construction is near-optimal is presented in \Cref{sec:lower-bound}. For improved readability, some proofs are only sketched out, with the full version deferred to \Cref{sec:deferred}. \Cref{sec:negativeassociation} covers the results on negatively associated random variables needed for the proof of the main Lemma.

\section{Technical Discussion}\label{sec:technical-discussion}
In this section, we give a technical overview of our contribution. However, before stating them we discuss some important concepts (that clarify the technical choices), and then state our results.

\subsection{Overview of Concepts}\label{sec:our-concepts}

A semi-oblivious routing over $G$ is a simple combinatorial object: a \emph{path system} $\{ P(s,t) \}_{s, t \in V}$, where each pair of vertices $(s, t)$ is associated with a collection $P(s, t)$ representing the candidate paths between $s, t$.

\begin{definition}[Path System]\label{def:path-system}
A \textit{path system} $\mathcal{P} = \{P(s, t)\}_{s, t \in V}$ is a collection of sets $P(s, t)$ of simple paths with endpoints $s$ and $t$, for every vertex pair $(s, t)$. We say a path system $\mathcal{P}$ is \textit{$\alpha$-sparse} if $|P(s, t)| \leq \alpha$ for all $(s, t)$. With slight abuse of notation, we say a path system $\mathcal{P}$ is \textit{$(\alpha + \cut)$-sparse} if $|P(s, t)| \leq \alpha + \cut(s, t)$, where $\cut(s, t)$ is the minimum cut between $s$ and $t$.
\end{definition}

\textbf{Competitive Ratio.} To evaluate the quality of our semi-oblivious routing, we perform the following sequence of stages.
\begin{enumerate}[label=\textit{(Stage \arabic*)}, leftmargin=*]\setlength\itemsep{0.15em}  % uses package {enumitem}
\item A graph $G = (V, E)$ is given as input.
\item We design a path system $\mathcal{P} = \{P(s, t)\}_{s, t \in V}$.
\item \label{stage:revelation} An arbitrary (possibly adversarially-chosen) demand is revealed (i.e., the multiset of source-destination pairs corresponding to packets).
\item \label{stage:assignment} For each packet $s \to t$ (with source $s$ and destination $t$) we choose which of the candidate paths $P(s, t)$ the packet $s \to t$ uses. We are allowed to adaptively use all available global information in a way to minimize the maximum edge congestion (i.e., number of paths going over any single edge).
\item \label{stage:comparions} Finally, the maximum edge congestion obtained in this way is compared to the offline optimal one and the ratio is called the \emph{competitive ratio}, the primary quantity we aim to minimize. If $\mathcal{P}$ has competitive ratio at most $C$ against all demands, we say $\mathcal{P}$ is a $C$-competitive semi-oblivious routing.
\end{enumerate}

%\alert{double-check!} This leads us to make a choice: whether we consider \emph{integral} or \emph{fractional} routings. In this section, we will present results for integral routings, while the results and proofs in the technical sections will focus on the (more general) fractional routings.

One can consider many types of (semi)oblivious routings. For example, we can construct routings that are either integral or fractional routings, we can compare ourselves to the optimal integral or fractional solution, we can look at demands that are either arbitrary or $\{0, 1\}$ or permutation (see below), consider $\alpha$-sparsity vs. $(\alpha + \cut)$-sparsity (\Cref{def:path-system}), etc. Most of these choices are inconsequential: one can typically inter-reduce results between them with negligible losses. However, some choices are incompatible and keeping track of them introduces technical complexities in the formal statements. Due to this, we explain these choices below and describe which combinations are meaningful.

\textbf{Fractional vs. Integral Routing.} When talking about fractional routings, in \ref{stage:assignment}, for each packet $s \to t$ we adaptively assign to each candidate path $p \in P(s, t)$ a nonnegative weight $w(p) \ge 0$ such that $\sum_{p \in P(s, t)} w(p) = 1$ (equivalently, we choose a distribution over $P(s, t)$). These weights define a \emph{fractional unit flow} from $s$ to $t$ which we use to route the packet. Similarly, the optimal fractional solution, to which we compare ourselves in \ref{stage:comparions}, routes each packet $s \to t$ using a (fractional unit) flow (i.e., via a convex combination of paths from $s$ to $t$). On the other hand, in integral routings, each packet $s \to t$ is routed on exactly one path from $P(s, t)$. In this case the optimal solution is restricted to be integral (as otherwise no sublinear competitive ratio would be obtainable).

\textbf{Types of Demands.} We also define several kinds of demands.
\begin{definition}\label{def:demands}
  A \textit{demand} is a function $d : V \times V \mapsto \mathbb{R}_{\geq 0}$ from vertex pairs to nonnegative real numbers, such that $d(v, v) = 0$ for all $v \in V$. The demand is \textit{integral} if $d(s, t)$ is an integer for every $(s, t)$. The demand is a \textit{\zeroonedemand} if $d(s, t) \in \{0, 1\}$. The demand is a \textit{permutation demand} if it is a \zeroonedemand and $\sum_{s} d(s, t) \leq 1$, $\sum_{t} d(s, t) \leq 1$. We denote the \textit{support} of a demand by $\supp(d) = \{(s, t) : d(s, t) > 0\}$ and define the \textit{size} $\siz(d)$ of a demand as $\siz(d) := \sum_{s \neq t} d(s, t)$.
\end{definition}
Permutation demands will mostly be used to give a technical overview and in the lower bound (since it makes the statement stronger).

\textbf{Meaningful Combinations.} In this paper we mostly consider two settings:
\begin{enumerate}[label=(\arabic*)]
\item Fractional routings, arbitrary demand, and $(\alpha + \cut)$-sparse path systems.
\item Integral routings, \zeroonedemands, and $\alpha$-sparse path systems.
\end{enumerate}
Other combinations are either weaker or not meaningful (in the sense that no sublinear competitive ratio is achievable).
For example, using $\alpha$-sparsity in Setting (1) is not meaningful as we need at least $\cut(s, t)$ many candidate paths between $s$ and $t$. To see why, consider two $n$-cliques connected via $n$ edges. The demand is comprised of a single packet going from an arbitrary vertex $s$ in one clique to an arbitrary $t$ in the other one: we need at least $\cut(s, t)/\beta = n/\beta$ many paths to be $\beta$-competitive since the optimal solution has congestion $1/n$ and the semi-oblivious congestion is at least $1/|P(s,t)|$. Finally, results for Setting (2) are implied by Setting (1).

\textbf{Logarithmic- vs. Low-Sparsity Cases.} We also consider cases when $\alpha$ is small and give the full sparsity-competitiveness trade-off for those cases. However, the case $\alpha = \bigO{\frac{\log n}{\log \log n}}$ and $\poly\log n$ competitiveness is of special interest and we refer to it as the \emph{logarithmic-sparsity case}. On the other hand, the general case (when $\alpha$ can be small) is referred to as the \emph{low-sparsity case}. We note that our results here show an exponential improvement in competitiveness with $\alpha$. Moreover, we nearly match the entire trade-off curve with a lower bound.

\textbf{Objective: Edge Congestion vs. Completion Time.} The default objective throughout the paper is minimizing the maximum edge congestion, where the congestion of $e \in E$ is the (potentially fractional) number of packets routed through $e$. A compelling alternative is minimizing the completion time: given a set of paths $P = \{ p_1, \ldots, p_k \}$ used for routing packets (connecting the source-destination pairs), we aim to minimize ``\emph{dilation} + \emph{congestion}''. Here, congestion is defined as before, and dilation is defined as the maximum number of hops of any path, $\max_{p \in P} \hop(p)$. This roughly corresponds to minimizing the time until the last packet arrives at their destination (due to classic reductions~\cite{leighton1994packet,GH16}), hence the name completion time. However, due to recent advancements in hop-constrained oblivious routings~\cite{GhaffariHZ21}, our results gracefully extend to completion-time-competitive semi-oblivious routings. Namely, we simply sample from a hop-constrained oblivious routing instead of the classic one.

% fractional routing, fractional opt, arbitrary demands, alpha+cut sparsity
% integral routing, integral opt, {0,1} demands
% (lower bound) permutation demand

%\textbf{Integral routings and $\{0, 1\}$-demands.}  If we also restrict the demand to be $\{0, 1\}$ (i.e., at most one packet is being routed between any pair of vertices), then we can prove $\tildeO(n^{\bigO{1/\alpha}})$-competitiveness with only $|P(s,t)| \le \alpha$ many candidate paths. It is easy to see that the $\{0, 1\}$ restriction is needed: an integral routing with a large demand is essentially a scaled fractional routing (the same example from above with a large demand is a counterexample).

\subsection{Formal Overview of Results for Integral Routings}\label{sec:formal-results}

% Simplify this paragraph, no need to be so formal with the formal definition right after

In this section, we present our results in the setting of integral routings and demands. For brevity, we omit the analogous results for fractional routings, albeit they directly follow from the results of \Cref{sec:sparse-semi-oblivious-routing}.

\textbf{Logarithmic-Sparsity Case.} We show that in every graph, there exists a set of logarithmically many paths between every vertex pair, through which any \zeroonedemand can be routed semi-obliviously with polylogarithmic competitiveness. Moreover, we give a matching lower bound that shows that polylogarithmic competitiveness is unachievable with asymptotically less paths, even with the more restricted demand set of permutation demands.

\begin{restatable}{theorem}{basiccorollary}\label{thm:basic}
Let $G$ be a $n$-vertex graph with at most a polynomial number of edges. Then, there exists a $\bigO{\log n / \log\log n}$-sparse integral semi-oblivious routing on $G$ that is $\bigO{\log^3 n / \log \log n}$-competitive on \zeroonedemands.
\end{restatable}

\begin{restatable}{lemma}{basiclowerboundcorollary}[Impossibility]\label{lem:basiclowerbound}
There exists an infinite family $\mathcal{G}$ of simple graphs, such that for any sparsity bound $g_1 = o\left(\frac{\log n}{\log \log n}\right)$ and competitiveness bound $g_2 = \poly \log n$, there exists an integer $n_0$ such that for every $n$-vertex graph $G \in \mathcal{G}$ where $n \geq n_0$, there exists no $g_1(n)$-sparse integral semi-oblivious routing on $G$ that is $g_2(n)$-competitive on all permutations demands.
\end{restatable}

\textbf{Low-Sparsity Case.} We now turn our attention to semi-oblivious routings with a sublogarithmic number of paths between every vertex pair. We show that any polynomial competitiveness is achievable with a constant number of sampled paths, and that every additional path yields a polynomial improvement to the competitiveness. The matching lower bound shows that this is tight: for any $\alpha$, the sparsity of \Cref{thm:basicalpha} cannot be improved by more than a constant factor while retaining the same competitiveness.

% For any $\alpha$, 

\begin{restatable}{theorem}{basicalphacorollary}\label{thm:basicalpha}
Let $G$ be a $n$-vertex graph with at most a polynomial number of edges. Then, for every positive integer $\alpha = o\left(\frac{\log n}{\log \log n}\right)$, there exists an $\alpha$-sparse integral semi-oblivious routing on $G$ that is $n^{\bigO{\alpha^{-1}}}$-competitive on \zeroonedemands.
\end{restatable}

\begin{restatable}{lemma}{basicalphalowerboundcorollary}[Impossibility]\label{lem:basicalphalowerbound}
There exists an infinite family $\mathcal{G}$ of simple graphs, such that for every constant $\epsilon > 0$ and $\alpha = o(\frac{\log n}{\log \log n})$, there exists an integer $n_0$ such that for every $n$-vertex graph $G \in \mathcal{G}$ where $n \geq n_0$, there exists no $\alpha$-sparse integral semi-oblivious routing on $G$ that is $n^{\left(\frac{1}{2} - \epsilon\right) \alpha^{-1}}$-competitive on all permutation demands.
\end{restatable}

\textbf{Arbitrary Integral Demands.} Our results generalize to arbitrary integral demands. However, this requires us to consider $(\alpha + \cut)$-sparsity, as no polylogarithmic competitiveness would be possible, and to pay an additional logarithmic factor in the competitiveness. We could derive an analogous result for the low-sparsity case, but omit it for brevity. %For the result on general demands, note that any $\alpha$-sparse integral semi-oblivious routing can be at best $\alpha^{-1} \max_{s, t} \cut(s, t)$-competitive. Thus, the size of cuts has to appear either in competitiveness or in sparsity.
\begin{restatable}{lemma}{cutcorollary}\label{lem:cut}
Let $G$ be a $n$-vertex graph with at most a polynomial number of edges. Then, there exists a $\left(\bigO{\frac{\log n}{\log \log n}} + \cut\right)$-sparse $\bigO{\frac{\log^4 n}{\log \log n}}$-competitive integral semi-oblivious routing.
\end{restatable}

\textbf{Routings for Completion Time.} Recall that ``completion time'' objective corresponds to minimizing $\mathrm{congestion} + \mathrm{dilation}$ (where the dilation is the longest path used in the routing). By combining our results with the recent hop-constrained oblivious routings~\cite{GhaffariHZ21}, we construct sparse semi-oblivious routings that are competitive in terms of completion time. Note that in this setting, compared to competitiveness in congestion alone, we require quadratic sparsity. One can also derive completion-time results against arbitrary demands, but we omit them for brevity.

%\begin{restatable}{lemma}{congdilationcorollary}\label{lem:congdilation}
%Let $G$ be a $n$-vertex graph with polynomially-bounded edge capacities. Then, there exists a $\bigO{\left(\frac{\log n}{\log \log n}\right)^2}$-sparse integral semi-oblivious routing $\calP$ on $G$, such that for any \zeroonedemand $d$ and routing $R$ that is integral on $d$, there exists a routing $R'$ on $\calP$ that is integral on $d$ such that both $\obl(R', d) \leq \obl(R, d)\ \poly\log n$ and $\dil(R', d) \leq \dil(R, d)\ \poly\log n$.
%\end{restatable}

\begin{restatable}{lemma}{congdilationcorollary}\label{lem:congdilation}
Let $G$ be a $n$-vertex graph with at most a polynomial number of edges. Then, there exists a $\bigO{(\log n)^2/(\log \log n)^2}$-sparse integral semi-oblivious routing $\calP$ on $G$ that is $\poly(\log n)$-completion time competitive on \zeroonedemands.
\end{restatable}

\begin{restatable}{lemma}{congdilationsparsecorollary}\label{lem:congdilationsparse}
Let $G$ be a $n$-vertex graph with at most a polynomial number of edges. Then, for every integer $\alpha = o(\log n / \log \log n)$, there exists an $\alpha^2$-sparse integral semi-oblivious routing $\calP$ on $G$ that is $n^{\bigO{\alpha^{-1}}}$-completion time competitive on \zeroonedemands.
\end{restatable}

%\antti{Say you want to approximate congestion + dilation. Then, you can pick $\alpha = o(\frac{\log n}{\log \log n})$. You want to get within $n^{h \alpha^{-1}}$ in terms of congestion. Since $\log_{n^{h \alpha^{-1}}}(n) = \bigO{\alpha}$, we'll have sparsity $\alpha^2$ for competitiveness $n^{\bigO{\alpha^{-1}}}$}

%\antti{There is a $\alpha^2 + \alpha \cut$ sparse $n^{\bigO{\alpha^-1}}$-competitive semi-oblivious fractional routing. You might be able to achieve $\alpha^2 + \cut$ somehow, doesn't seem trivial (separate cut into different length paths?).}

\section{Related Work}\label{sec:related-work} 

\textbf{Routing on Hypercubes.} Oblivious routing was first studied by Valiant and Brebner~\cite{ValiantB81} for the special case of the hypercube as the underlying graph. This is because parallel computers are often implemented using (a variant of) a hypercube architecture as the topology of choice for connecting its processors, motivating the study of oblivious routings. A very simple strategy known as the ``Valiant's trick'' yields a $(\poly\log n)$-competitive routing (in terms of edge congestion): when routing a packet from $s_i \to t_i$, first greedily route it from $s_i$ to a random intermediate vertex, and then greedily route it to $t_i$. However, this routing is inherently randomized, and \cite{BorodinH85} and \cite{KaklamanisKT91} show that any deterministic routing on a hypercube cannot be $\tilde{o}(\sqrt{n})$-competitive. Specifically, they show the following slightly-more-general result: in any $n$-vertex graph $G$ with maximum degree $\Delta$, for any deterministic oblivious routing $R$, there exists a permutation demand $d$ such that routing it yields congestion at least $\Omega\left(\sqrt{n} / \Delta\right)$. Oblivious routings on a range of special graphs were also studied, including expanders, Caley graphs, fat trees, meshes, etc.~\cite{upfal1984efficient,rabin1989efficient,scheideler1998universal,busch2008optimal,busch2010optimal}.

\textbf{Oblivious Routing on General Graphs.} Raecke~\cite{Racke02} first demonstrated that in every graph there exists a $(\poly\log n)$-compe\allowbreak titive oblivious routing. No poly-time construction algorithm was known at the time, hence \cite{BienkowskiKR03}, \cite{HarrelsonHR03}, and \cite{Madry10} gave polynomial time construction algorithms for the \textit{hierarchical decomposition} required by \cite{Racke02}. This line of work culminated in the celebrated result of Raecke~\cite{Racke08} who gave a poly-time construction of the $\mathcal{O}(\log n)$-competitive oblivious routing scheme by reducing the problem to $\bigO{\log n}$-distortion tree embeddings. The $\bigO{\log n}$ is asymptotically optimal on general graphs, as shown by \cite{BartalL97} and \cite{MaggsHVW97}. The first close-to-linear-time construction (at a cost of polylogarithmically-higher competitiveness) was given by \cite{RackeST14} by constructing a hierarchical expander decomposition.

%\cite{AzarCFKR03}: LP for optimal oblivious routing. %% % @@@ @IDK what to do with this????

\textbf{Other Quality Measures.} All aforementioned oblivious routings had the objective of minimizing (maximum) edge congestion and were optimized either for competitiveness or runtime. However, other quality measures are also prominent in the literature. Various $\ell_p$ norms of edge congestions and generalizations were studied~\cite{englert2009oblivious, gupta2006oblivious,Racke09}. Notably, very recent work has shown that there exists a $(\poly\log n)$-competitive oblivious routing for congestion+dilation~\cite{GhaffariHZ21}, and \cite{GHR21} has shown that these structures can often be constructed in almost-optimal time in the sequential, parallel, and distributed settings. Furthermore, starting with \cite{Racke019}, considerable effort has been invested in designing oblivious routing schemes which are \emph{compact}, meaning that the size of the routing tables is small. \cite{CzernerR20} gave compact oblivious routing schemes in weighted graphs where the hop-lengths (i.e., dilation) of the returned paths are not controlled for, while \cite{GHR21} gave compact oblivious routing schemes which control for both congestion and dilation.

\textbf{Semi-Oblivious Routing.} The extension from oblivious to semi-oblivious routing was proposed in \cite{HajiaghayiKL07} to model issues that naturally arise in VLSI design and traffic engineering. However, their results were negative and focus mostly on lower bounds: they excluded the possibility of polynomially-sparse semi-oblivious routing with $\mathcal{O}(1)$-competitiveness by arguing that every such routing on $n \times n$ grids cannot be better than $\Omega\left(\frac{\log n}{\log \log n}\right)$-competitive. \cite{KumarYYFKLLS18} present a practical implementation of a semi-oblivious routing-based algorithm for the traffic engineering problem. They empirically show that $\alpha$-sparse semi-oblivious routing offers near-optimal performance and satisfactory robustness even for small constant $\alpha$ (e.g., $\alpha = 4$).

% Existing results on semi-oblivious routings focus on lower bounds. \cite{HajiaghayiKL07} show that no polynomially sparse semi-oblivious routing can be $\bigO{\log^{1 - \epsilon} n}$-competitive, and applying a result of \cite{KaklamanisKT91} to hypercubes, where any permutation demand can be routed with constant congestion, shows that a $1$-sparse semi-oblivious routing cannot be $\smallo{\frac{\sqrt{n}}{\log n}}$-competitive. 

% \begin{theorem}[Theorem 3.3 of \cite{HajiaghayiKL07}]\label{thm:polynomialsparsitylowerbound}
% There exists series-parallel graphs of size $n$ in which every polynomially sparse semi-oblivious routing has competitive ratio $\Omega\left(\frac{\log n}{\log \log n}\right)$.
% \end{theorem}

% \cite{ThorupZ01}: Compact oblivious routing schemes with short paths. However, no guarantees on competitiveness? There is a lot of research on this type of oblivious routings.  % Goran: these are not typically referred to as "oblivious routing" as they don't control for congestion, hence the obliviousness between packets is trivial ... they are typically called "distance oracles"

%\cite{HajiaghayiKL07} show that there exist graphs where every polynomially sparse semi-oblivious routing has competitiveness $\Omega\left(\frac{\log n}{\log \log n}\right)$. In particular, they show this for $n \times n$ grids and series-parallel graphs, which have treewidth 2.

\section{Formal Notation}\label{sec:notation}

\textbf{Graphs.} We denote undirected graphs with $G = (V, E)$. In place of capacities, we allow $E$ to contain parallel edges. We only work with undirected and connected graphs with polynomially many edges, and won't state that a graph is undirected and connected every time we declare one. We write $n := |V|$ and $m := |E|$.

\textbf{Minimum Cut.} $\cut(s, t) : V \times V \mapsto \mathbb{Z}_{\ge 0}$ denotes the size of the minimum $(s, t)$-cut $G$. We define $\cut(v, v) = 0$.

\textbf{Routings.} A \textit{routing} $R = \{R(s, t)\}_{s, t \in V}$ is a collection of distributions $R(s, t)$ over simple $(s, t)$-paths for every vertex pair $(s, t)$. In other words, a routing is a path system that also assigns weights to paths. A routing $R$ ``routes a demand $d$'' by assigning a weight $d(s, t) \prob{ R(s,t) = p }$ for every path $p \in \supp(R(s, t))$. We say a routing $R$ is \textit{integral} on a demand $d$ if $d(s, t) \prob{R(s,t) = p}$ is an integer for every $s, t, p$. We define the support $\supp(R(s, t))$ of $R(s, t)$ to be the set of $(s, t)$-paths $R(s, t)$ assigns nonzero probability to and the support $\supp(R)$ of $R$ to be the path system $\calP$ where $P(s, t) = \supp(R(s, t))$. We say that a routing $R$ is supported \textit{on a path system} $\mathcal{P}$ if $\supp(R(s, t)) \subseteq P(s, t)$ for all $(s, t)$.

\textbf{Congestion.} The \textit{congestion} of an edge $e$ on a routing $R$ and demand $d$ is the sum of weights of paths using that edge: $\obl(R, d, e) = \sum_{s, t} d(s, t) \prob{e \in R(s, t)}$. The \textit{congestion of a routing} $R$ on a demand $d$ is the highest edge congestion $\obl(R, d) = \max_{e \in E} \obl(R, d, e)$.

\textbf{Dilation.} We denote by $\hop(p)$ the \textit{hop-length} of $p$, or the number of edges in $p$. The \textit{dilation} $\dil(R, d) = \allowbreak \max_{(s, t) \in \supp(d)} \allowbreak \max_{p \in \supp(R(s, t))} \allowbreak \hop(p)$ of a routing $R$ of $d$ is the maximum hop-length over paths that $R$ assigns a positive weight to on $d$.

\textbf{Optimal Congestion.} For a demand $d$, the \textit{optimal congestion} $\opt(d) = \min_{R} \obl(R, d)$ is the minimum congestion over routings $R$ of $d$. For integral $d$, the \textit{optimal integral congestion} $\optint(d) = \min_{\text{$R$ integral on $d$}} \obl(R, d)$ is the minimum congestion over routings that are integral on $d$.

\textbf{Oblivious Routings.} An \textit{oblivious routing} is a routing $R$. We say $R$ is \textit{$C$-competitive on a set of demands $D$} for $C \geq 1$ if, for every demand $d \in D$, we have $\obl(R, d) \leq C \opt(d)$. We say $R$ is \textit{$C$-competitive} for $C \geq 1$ if it is $C$-competitive on the set of all demands.

\textbf{Integer Prefix Sets.} We use the notation $[n] = \{1, \dots, n\}$.

\textbf{Logarithms.} We denote the base-$2$ logarithm by $\log$ and the base-$e$ logarithm by $\ln$. 

\section{Constructing Sparse Fractional Semi-Oblivious Routings}\label{sec:sparse-semi-oblivious-routing}

In this section, we prove the competitiveness of fractional $(\alpha + \cut)$-sparse semi-oblivious routings sampled from an oblivious routing.

The section starts by formally defining semi-oblivious routings and $\alpha$-samples, then states our main Theorem. In \Cref{sec:technical-overview}, we give a brief overview of the proof of the main Theorem. The proof is divided into the description of the organisation of the proof and the definitions required for it  in \Cref{sec:proof-structure}, a proof of the main Lemma in \Cref{sec:proving-the-main-lemma} and the reduction from the main theorem to the main lemma is finished in \Cref{sec:finishing-the-reduction}.

We start by formally defining semi-oblivious routings. On its face, a semi-oblivious routing is a path system $\calP$ (\Cref{def:path-system}) with additional competitiveness properties. To route a demand, the semi-oblivious routing picks the demand-dependent optimal routing supported on a subset of the paths in $\calP$. For example, if $\calP$ contains every simple $(s, t)$-path, the semi-oblivious routing is trivially $1$-competitive, thus we are interested in bounding both the competitiveness and the \textit{sparsity} of $\calP$: the maximum amount of paths between any $(s, t)$-pair.

% \begin{definition}[Path system]\label{def:path-system}
% A \textit{path system} $\mathcal{P} = \{P(s, t)\}_{s, t \in V}$ is a collection of sets $P(s, t)$ of simple $(s, t)$-paths for every vertex pair $(s, t)$. We say a path system $\mathcal{P}$ is \textit{$\alpha$-sparse} if $|P(s, t)| \leq \alpha$ for all $(s, t)$. With slight abuse of notation, we say a path system $\mathcal{P}$ is \textit{$(\alpha + \cut)$-sparse} if $|P(s, t)| \leq \alpha + \cut(s, t)$.
% \end{definition}

\begin{definition}[Semi-Oblivious Routing]

A \textit{semi-oblivious routing} is a path system $\calP$. For a demand $d$, we define the congestion of (fractionally) routing $d$ via $\calP$: \begin{equation*}
    \oblreal(\calP, d) = \min_{\text{$R$ a routing on $\calP$}} \obl(R, d).
\end{equation*}
\textbf{Competitiveness.} We say $\calP$ is \textit{$C$-competitive on a set of demands $D$} for $C \geq 1$ if for every demand $d \in D$ we have $\oblreal(\calP, d) \leq C \opt(d)$, and that $\calP$ is \textit{$C$-competitive} if it is $C$-competitive on the set of all demands.

\textbf{Competitiveness With an Oblivious Routing.} We say $\calP$ is \textit{$C$-competitive} with an oblivious routing $R$ \textit{on a set of demands $D$} for $C \geq 1$ if for every demand $d \in D$ we have $\oblreal(\calP, d) \leq C \obl(R, d)$. We say $\calP$ is $C$-competitive with $R$ if it is $C$-competitive with $R$ on the set of all demands.
\end{definition}

% \alert{move this too?} The definition of $(\alpha + \cut)$-sparsity is motivated by the fact that an $\alpha$-sparse semi-oblivious routing can at best be $\max_{s, t} \frac{\cut(s, t)}{\alpha}$-competitive: consider a demand of size $1$ where $d(s, t) = 1$. The optimal congestion is $\frac{1}{\cut(s, t)}$, but the semi-oblivious routing cannot achieve congestion less than $\frac{1}{\alpha}$. Such a counterexample doesn't exist for integral semi-oblivious routings on \zeroonedemands, where we obtain polylogarithmic competitiveness with sublogarithmic sparsity.

% \alert{G: where should we move this? this seems not the right place}
% \alert{G: can we cut stuff from here???}

% \cite{HajiaghayiKL07} also\alert{! stuff moved} note that the support of a $C$-competitive oblivious routing is a $C$-competitive semi-oblivious routing, thus in particular, polynomially sparse $\bigO{\log n}$-competitive semi-oblivious routings exist. However, note that sparse semi-oblivious routings cannot be obtained with this method, as by Theorem \ref{thm:deterministicobliviouslowerbound}, any oblivious routing with $\alpha$-sparse support can at best be $\frac{\sqrt{n}}{\alpha \log n}$-competitive on a hypercube.

Our methods of constructing semi-oblivious routings are based on sampling paths from an oblivious routing, instead of taking its whole support. We define samples of an oblivious routing as follows:

\begin{definition}[Sample of an Oblivious Routing]

Let $G$ be a graph, $R$ be an oblivious routing and $\alpha \in \mathbb{Z}_{> 0}$ be a parameter. An $\alpha$-sample of $R$ is a semi-oblivious routing $\calP$, where $P(s, t) \subseteq \supp(R(s, t))$ is a set of $\alpha$ paths sampled with replacement from the probability distribution $R(s, t)$. An $(\alpha + \cut)$-sample is defined similarly, but with $(\alpha + \cut(s, t))$ sampled paths between $s$ and $t$.
\end{definition}

The following theorem that we prove in the next section states that an $(\alpha + \cut)$-sample of an oblivious routing is competitive with the oblivious routing with high probability\footnote{We use the phrase ``with high probability'' to mean that a claim holds with probability at least $1 - n^{-C}$, where $C > 0$ is any sufficiently large constant. In other words, other constants in the statement can be tuned to account for any sufficiently large $C$.}. The second part of the theorem allows us to cut a logarithmic factor from the competitiveness of integral semi-oblivious routings, as the results in that setting are achieved through a reduction where the special condition applies.

We note that it is convenient to measure the competitiveness of a semi-oblivious routing against the result provided by a (standard) oblivious routing, instead of competitiveness again offline solutions. This allows us to apply our results in a black-box manner to the oblivious routings of \cite{GhaffariHZ21} to achieve completion-time competitive semi-oblivious routings. %The special type of demands described where one $\log$-factor can be cut will be relevant for the reduction we do to prove the competitiveness of $\alpha$-samples.

\begin{restatable}[$(\alpha + \cut)$-Sample Theorem]{theorem}{mainthm}\label{thm:mainthm}
    Let $G = (V, E)$ be a $n$-vertex graph with polynomially many edges, $R$ be an oblivious routing and $\alpha \in [n]$ be a parameter. Let $\calP$ be an $(\alpha + \cut)$-sample of $R$. Then, with high probability, $\calP$ is 
    \begin{equation*}
        \bigO{\log^2(n) \left(\alpha + n^{\bigO{\alpha^{-1}}}\right)}
    \end{equation*}
    -competitive with $R$, and for every subset $D$ of \zeroonedemands such that for $d \in D$, we have $\supp(d) \subseteq \{(s, t) : \cut(s, t) = 1\}$, $\calP$ is
    \begin{equation*}
        \bigO{\log(n) \left(\alpha + n^{\bigO{\alpha^{-1}}}\right)}
    \end{equation*}
    -competitive with $R$ on $D$.
\end{restatable}

\subsection{Technical Overview}\label{sec:technical-overview}

\textbf{The Hypercube/Permutation-Demand Subcase.} For simplicity, in this section we overview the proof of a weaker version of \Cref{thm:mainthm} for the case of permutation demands on a hypercube. We show that in a hypercube, with high probability, sampling $\alpha = \Theta(\log n)$ paths between every pair of vertices produces a semi-oblivious routing that can route any permutation demand with congestion $\poly \log n$. This subcase contains all the main ideas, while avoiding the intricacies involved with general graphs, lower sparsities, and arbitrary demands.

We note that there exists a simple oblivious routing $R$ for hypercubes: for each packet, we (obliviously) uniformly randomly choose an intermediate node, and then greedily route the packet from the source to the intermediate node, and then to the destination~\cite{ValiantB81}. This strategy ensures that, for any permutation demand, the (expected) congestion of any edge is $\bigO{1}$.

Let $\alpha = \Theta(\log n)$. Let $\calP$ be a semi-oblivious routing constructed by, for every pair $s, t \in V$, taking $\alpha$ samples of $(s, t)$-paths from the oblivious routing $R$. We show that, with high probability, for every permutation demand, we can (adaptively) choose at least one path of the $\alpha$ sampled paths for each packet, such that the chosen paths yield congestion of at most $h = \poly \log n$.

\textbf{Strategy: Exponentially-Small Failure Per Demand.} Our proof strategy that $\calP$ can (semi-obliviously) route every permutation demand is to show the following: For any fixed permutation demand $d$, $R$ can route $d$ with congestion $h$ except with failure probability that is exponential in the size of the demand (namely  $n^{-\Omega(\siz(d))}$). Then, the proof can be finished with a union bound over permutation demands, as the number of permutation demands of size $s$ is at most exponential in $s$ (namely $n^{2s}$). Hence, with high probability, $R$ works for all permutation demands.

We note that there are many standard arguments that show that a fixed demand fails with probability at most $\exp(- \poly \log n)$ (such as randomized rounding~\cite{raghavan1987randomized}). In fact, this requires no semi-obliviousness. However, this is not sufficient as there are exponentially many permutation demands (namely $n! = \exp(\Theta(n \log n))$). A significantly more intricate argument is required.

\textbf{Idea: Weak Routing.} It is sufficient to show that the routing $\calP$ can route \emph{at least half} of any permutation demand with congestion $h$. Any semi-oblivious routing with that property can be shown to be able to route any permutation demand with congestion $h \cdot \Theta(\log n) = \poly \log n$, by repeatedly routing half of the demand. Thus, it suffices to show that for any fixed permutation demand $d$, with failure probability $n^{-\Omega(\siz(d))}$, $\calP$ can route at least half of $d$ with congestion $h$.

\textbf{Idea: Dynamic Process.} Consider all the $\alpha \cdot \siz(d)$ candidate paths that the sampled $\calP$ can use to route $d$. If the congestion of every edge in $G$ is at most $h$, we are done, as selecting arbitrary paths to route the demand cannot achieve higher congestion than selecting every path. Otherwise, we find an overcongested edge $e$ (with congestion $t > h$) and delete all candidate paths that go through $e$. The crucial observation here is that the probability of $t > h$ paths overcongesting $e$ is at most $\exp(- \Omega(t))$. To explain why, consider a random variable $X$ that gives the number of sampled paths crossing a fixed edge $e$. As the base oblivious routing $R$ achieves constant congestion, the expectation of the random variable is at most $\E[X] \le \Theta(\alpha)$, so $t > h \geq 2E[X]$. Since the paths are sampled independently, a Chernoff bound ensures $\Pr[X > t] < \exp(- \Omega(t))$. We note that the variables are not actually independent, but can be proven to be negatively associated, allowing the use of the bound without changes~\cite{Kum83}.

We repeat the process for all overcongested edges and delete all congesting paths. After this, if at most half of the candidate paths were deleted, for at least half of the vertex pairs in the demand at least one path remains. Routing the demand between those pairs through those paths routes at least half of the demand with congestion at most $h$, as desired.

An uninitiated reader could assume we are done: let the congestions of the overcongested edges be $t_1, \dots, t_j$. Assuming independence, the probability of these overcongestions is $\exp(-\Omega(\sum_i t_i)) \le n^{- \Omega(\siz(d))}$, where the last inequality used the fact that weak routing fails, i.e., $\sum_i t_i > \alpha \cdot \siz(d) / 2 = \Omega(\siz(d) \log n)$. However, a big issue is that overcongestions are not independent. A more intricate approach needs to be taken.

\textbf{Idea: Bad Patterns.} To resolve the inter-dependencies, we first fix an arbitrary ordering of edges $(e_1, e_2, \ldots, e_m)$ that is independent of $\calP$ or $d$. The dynamic process will iterate over the edges in the fixed order and, if the current edge is overcongested, delete all candidate paths congesting the edge. To analyze the number of paths removed, we define \textit{bad patterns}: $m$-tuples of integers $(t_1, \ldots, t_m) $ with $\sum_i t_i > \alpha\cdot \siz(d)/2$. A bad pattern describes the path deletions during a failed process. During the processing of a demand, we say that the bad pattern occurs, if for every edge $e_i$, we deleted at least $t_i$ paths while processing $e_i$. If more than half of the candidate paths are deleted during the process, at least one bad pattern occurs.

Since we never delete more than zero but at most $h$ paths during a step, we require that every nonzero $t_i$ is greater than $h$. Thus, there are at most $\alpha \cdot \siz(d) / h$ nonzero values in the bad pattern, and with simple combinatorial calculations we can bound the number of bad patterns to $n^{\bigO{\siz(d)}}$. Thus, it suffices to show that any fixed bad pattern occurs with probability at most $n^{-\Omega(\siz(d))}$, at which point a union bound finishes the proof.

\textbf{Resolving Dependence.} We note that if a bad pattern occurs, in the initial set of sampled paths there must be at least $t_i$ paths crossing edge $e_i$ that \emph{do not} cross any overcongested edge $e_j$ that is ``earlier'' ($j < i$). Those paths would already have been eliminated while processing edge $e_j$ before the process reaches $e_i$. Thus, after fixing a potential bad pattern $(t_1, \ldots, t_m)$, we can uniquely assign each path $p$ to the first edge $e_i$ (in the ordering) on $p$ with $t_i > h$. Then, if a bad pattern occurs, $\calP$ must have sampled at least $t_i$ paths assigned to $e_i$, the probability of which is $\exp(-\Omega(t_i))$ as before via a Chernoff bound on negatively associated variables. Crucially, since the assignment is unique and does not depend on the process, these individual probabilities can be multiplied together. More formally, the probability that multiple lower bounds on disjoint subset sums of negatively associated variables simultaneously occur is at most the product of the probabilities of the individual lower bounds~\cite{Kum83}. Hence, the probability of a bad pattern is $\exp(-\Omega(\sum_i t_i))$ and we can conclude the same way as before.

\textbf{The General Case.} Recall that in the case of general graphs, it is necessary to sample at least $\cut(s, t)$-many paths between $s$ and $t$ to achieve competitiveness (as explained in \Cref{sec:our-concepts}). Because the number of sampled paths can now vary, we need to weight the sampled $(s, t)$-paths by dividing the $(s, t)$-demand equally among them, instead of setting the weight of every path to $1$. 

\textbf{Idea: Special Demands.} However, this causes another issue. Consider the random variables whose sum we bound via Chernoff in the hypercube case. Unlike before, their scales vary wildly, which significantly degrades the concentration of their sum. To resolve this, it is natural to work with \textit{special demands}: demands where $d(s, t)$ is either $0$ or $\alpha + \cut(s,t)$. Special demands force the random variables comprising the considered sum to be binary. This re-enables Chernoff and fixes the proof, at least for the case of special demands. However, this is enough --- we can show that a semi-oblivious routing that is competitive on special demands is competitive on general due to simple bucketing (up to a single logarithmic factor).

Another small issue is that the extra paths grow the set of bad patterns too large. To handle this, we force $t_i$ to be a multiple of $h$, and halve the sum requirement. Then, there still exists at least one bad pattern that occurs whenever more than half of the paths get cut, while the total number of bad patterns is sharply cut.

\textbf{The Low-Sparsity Case.} To achieve results for $\alpha \ll \log n$, we need to use the large deviation version of Chernoff bounds, providing the stronger bound of $\Pr[X > t] < \exp(-\Omega(t \log (t / \mathbb{E}[X])))$ when $t \gg \mathbb{E}[X]$. The extra logarithmic term also allows us to cut the required sparsity to achieve polylogarithmic competitiveness from $\Omega(\log n)$ to $\bigO{\log(n) / \log \log n}$.

\subsection{Structure of the Proof}\label{sec:proof-structure}

In this section, we give a short overview of the formal proof of \Cref{thm:mainthm}. The most important part of the proof is \Cref{lem:samplingorder}, which is a weaker version of the result: for every fixed demand of a special class of demands we define, with failure probability exponential in the support size of the demand, a sample independent of the demand can route at least half of the total demand with low congestion. Then, a union bound over the special class of demands in the proof of \Cref{cor:weakspecialdemands} shows that a sample can route at least half of every special demand with low congestion with high probability, and two further reductions generalise from routing half of the demand to routing the full demand (\Cref{lem:weaktostrong}) and from special demands to arbitrary demands (\Cref{lem:specialtogeneral}).

The notion of \textit{weak-competitiveness} formalises the notion of competitively routing half of the total demand. Note that in particular we do not require the semi-oblivious routing to be able to route the subdemand with congestion at most $C$ times the congestion of the oblivious routing on the subdemand, but instead compare against the congestion on the original demand. In \Cref{lem:weaktostrong}, we show that weak competitiveness implies competitiveness in general at the cost of one logarithmic factor.

\begin{definition}[Weakly-Competitive Semi-Oblivious Routing]\label{def:weakcompetitive}

For $C \geq 1$, a semi-oblivious routing $\calP$ is $C$-weakly-competitive with an oblivious routing $R$ on a set of demands $D$, if, for every demand $d \in D$, there exists a demand $d'$ and a routing $R'$ on $\calP$ such that
\begin{itemize}
    \item $d'(s, t) \leq d(s, t)$ for all $s, t$,
    \item $\siz(d') \geq \frac{1}{2} \siz(d)$,
    \item $\obl(R', d') \leq C \obl(R, d)$.
\end{itemize}
\end{definition}

The special class of demands contains the demands for which the ratio of $(s, t)$-demand to the number of $(s, t)$-paths $d(s, t) / (\alpha + \cut(s, t))$ is either $0$ or $1$. The definition of this class is necessary, as the proof of the main Lemma \ref{lem:samplingorder} requires the ratio to be constant for all nonzero demands to achieve concentration. We prove in \Cref{lem:specialtogeneral} that competitiveness on special demands implies competitiveness on general demands at the cost of an (multiplicative) logarithmic factor.

\begin{definition}[Special Demands]

Let $G$ be a $n$-vertex graph and $\alpha \in [n]$ be a fixed parameter. A demand $d$ is \textit{$\alpha$-special}, if $d(s, t) \in \{0, \alpha + \cut(s, t)\}$ for every $(s, t)$. We denote the set of all $\alpha$-special demands on $G$ by $\spec$. When $\alpha$ is clear from the context, we refer to $\alpha$-special demands as just "special demands".
\end{definition}

With these two definitions, we can formally state \Cref{lem:samplingorder}, the core part of our proof. The parameter $h$ that appears here explicitly can later be set to achieve the "with high probability" guarantee in \Cref{thm:mainthm}. Section \ref{sec:proving-the-main-lemma} is dedicated to the proof of \Cref{lem:samplingorder}.

\begin{restatable}[Main Lemma]{lemma}{samplingorderLemma}\label{lem:samplingorder}
Let $G = (V, E)$ be a $n$-vertex $m$-edge graph where $3 \leq n \leq m$, $\alpha \in [n]$ and $h \geq 1$ be fixed parameters, $R$ be an oblivious routing on $G$ and $d \in \spec$ be a fixed special demand. Let $\calP$ be an $(\alpha + \cut)$-sample of $R$. Then, with probability at least $1 - m^{-(h + 3) |\supp(d)|}$, $\calP$ is $\left(\alpha + m^{16(h + 7) \alpha^{-1}}\right)$-weakly-competitive with $R$ on $\{d\}$. 
\end{restatable}

The proof of \Cref{thm:mainthm} from \Cref{lem:samplingorder} is achieved through three reductions. First, a simple union bound shows that the probability bound of \Cref{lem:samplingorder} is sufficient for a random sample to be weakly competitive on all special demands, proving  \Cref{cor:weakspecialdemands}, a version of the main lemma with the correct sampling order. Then \Cref{lem:weaktostrong} reduces from weak competitiveness on special demands to general competitiveness on special demands, and finally \Cref{lem:specialtogeneral} shows competitiveness on general demands for any routing that is competitive on special demands. The proofs of the three reductions are covered in \Cref{sec:finishing-the-reduction}. Combining \Cref{cor:weakspecialdemands}, \Cref{lem:weaktostrong} and \Cref{lem:specialtogeneral}, the proof of \Cref{thm:mainthm} immediately follows.

\begin{restatable}{corollary}{weakspecialdemandCorollary}\label{cor:weakspecialdemands}
Let $G$ be a $n$-vertex $m$-edge graph where $3 \leq n \leq m$, $\alpha \in [n]$ and $h \geq 1$ be fixed parameters, and $R$ be an oblivious routing on $G$. Let $\calP$ be an $(\alpha + \cut)$-sample of $R$. Then, with probability at least $1 - m^{-h}$, $\calP$ is $\left(\alpha + m^{16(h + 7) \alpha^{-1}}\right)$-weakly competitive with $R$ on $\spec$. 
\end{restatable}

\begin{restatable}[weak-to-strong reduction]{lemma}{weaktostrongLemma}\label{lem:weaktostrong}
Let $G$ be a $n$-vertex $m$-edge graph, $R$ be an oblivious routing and $D$ be a set of demands, such that for every demand $d \in D$, for every demand $d'$ such that $d'(s, t) \in \{0, d(s, t)\}$, we have $d' \in D$. Let $\calP$ be a semi-oblivious routing that is $C$-weakly-competitive with $R$ on $D$. Then, $\calP$ is $\bigO{C \log m}$-competitive with $R$ on $D$.
\end{restatable}

\begin{restatable}[special-to-general reduction]{lemma}{specialtogeneralLemma}\label{lem:specialtogeneral}

Let $G$ be a $n$-vertex $m$-edge graph, $\alpha \in [n]$ be a fixed parameter, $R$ be an oblivious routing and $\calP$ be a semi-oblivious routing on $G$ that is $C$-competitive with $R$ on $\spec$. Then, $\calP$ is $\bigO{C \log(m)}$-competitive with $R$ on all demands.
\end{restatable}

\mainthm*

\begin{proof}
For $m < n$, since the graph has to be connected, there is exactly one simple path between any two vertices, thus $\calP$ is $1$-competitive. Thus, we may assume that $3 \leq n \leq m$, as for constant $n$ and $m$, no guarantees are made.

By \Cref{cor:weakspecialdemands}, with high probability, the $(\alpha + \cut)$-sample $\calP$ is $\left(\alpha + m^{\bigO{\alpha^{-1}}}\right)$-weakly competitive with $R$ on $\spec$. From now on, assume it is. Then, by \Cref{lem:weaktostrong}, $\calP$ is $\bigO{\log(m) \left(\alpha + m^{\bigO{\alpha^{-1}}}\right)}$-competitive with $R$ on $\spec$. Thus, since $D \subseteq \spec$, the sample $\calP$ is $\bigO{\log (m) \left(\alpha + m^{\bigO{\alpha^{-1}}}\right)}$-competitive with $R$ on $D$, and by \Cref{lem:specialtogeneral}, $\calP$ is $\bigO{\log^2 (m) \left(\alpha + m^{\bigO{\alpha^{-1}}}\right)}$-competitive w/ $R$. Finally, since $m$ is polynomially bounded, $\log m = \bigO{\log n}$ and $m^{\bigO{\alpha^{-1}}} = n^{\bigO{\alpha^{-1}}}$.
\end{proof}

\subsection{Proving the Main Lemma}\label{sec:proving-the-main-lemma}

The entirety of this section is dedicated to the proof of \Cref{lem:samplingorder}.

\samplingorderLemma*

Let $\cou(s, t) := \alpha + \cut(s, t)$ be the number of $(s, t)$-paths we sample, and let $D := \siz(d)$. Let $\gamma := \obl(R, d) \left(\alpha + m^{16(h + 7) \alpha^{-1}}\right)$ be the maximum allowed congestion of any edge. Index the edges arbitrarily, such that $E = \{e_1, \dots, e_{m}\}$.

We start by defining random variables $X$ that describe the $(\alpha + \cut)$-sampling of $\calP$ from $R$. Let $X(s, t)_{i, p}$ be a $\{0, 1\}$-random variable, with value $1$ if and only if the $i$th sampled $(s, t)$-path is $p$. Then, $p \in P(s, t)$ if and only if $X(s, t)_{i, p} = 1$ for some $i$. For fixed $s, t, i$, exactly one of $X(s, t)_{i, p}$ will equal $1$ and the probability that $X(s, t)_{i, p}$ equals $1$ is $\prob{R(s, t) = p}$. The variables $X(s, t)_{i, p}$ do not depend on $X(s', t')_{i', p'}$ with $i' \neq i$. Thus, by \Cref{lem:sumnegativeassociation} and \Cref{lem:indnegativeassociation}, the variables $X$ are \emph{negatively associated}.

We certify that the sampled semi-oblivious routing $\calP$ has the desired competitiveness with $R$ on $\{d\}$ by giving a process that finds weights $w_m$ for paths in $\calP$, such that, routing $w_m(p)$ of the $(s, t)$-demand along the $(s, t)$-path $p$ for all $(s, t, p)$, the congestion is at most $\gamma$ and at least half of the total demand is routed. The process can fail, and our goal is to prove that it fails with probability at most $m^{-(h+3)|\supp(d)|}$.

We find $w_m$ by defining a sequence of path weights $w_0, \dots, w_m$. The initial weight $w_0$ of a $(s, t)$-path for $(s, t) \in \supp(d)$ is simply the number of times it was sampled. As the demand is special, the initial sum of weights of $(s, t)$-paths then equals $d(s, t)$: using the initial weights all of the demand can be routed, but the resulting congestion might be arbitrarily high. The process iterates over the edges of the graph in the fixed order. While handling edge $e_k$, if the congestion applied to the edge by the weights $w_{k - 1}$ is greater than the allowed $\gamma$, the weight of all paths crossing the edge is set to $0$, and otherwise, all weights remain unchanged. Then, at the end of the process, the congestion of any edge is guaranteed to be at most $\gamma$, and it simply remains to bound the probability the final weights $w_m$ route at least half of the demand.

Formally, the weights $w_0, \dots, w_m$ are random functions, with $w_0$ defined as
\begin{equation*}
    w_0(p) := d(s, t) \frac{\sum_{i} X(s, t)_{i, p}}{\cou(s, t)} = \sum_{i} X(s, t)_{i, p} \mathbb{I}\left[(s, t) \in \supp(d)\right]
\end{equation*}
(as the demand is special), and for $k = 1, \dots, m$, the next weights $w_{k}$ being defined based on $w_{k - 1}$ based on the current congestion of edge $e_k$:
\begin{itemize}
    \item Let $\text{cong}_k := \sum_{p \in \calP} w_{k - 1}(p) \ind{e_k \in p}$ be the congestion of edge $e_k$ before step $k$.
    \item If $\text{cong}_k \leq \gamma$, we let $w_k := w_{k - 1}$. Otherwise, we let $w_{k}(p) := w_{k - 1}(p) \ind{e_k \not\in p}$.
\end{itemize}
Note that in both cases, $w_k(p) \leq w_{k - 1}(p)$.

With the final weights, the subset $d'$ of the demand we'll route is simply the sum of the final weights, $d'(s, t) = \sum_{p \in P(s, t)} w_{m}(p)$, and the routing $R'$ we use distributes load to the sampled paths according to their weights: $\prob{R'(s,t) = p} = \frac{w_{m}(p)}{d'(s, t)}$ (for $(s, t) \in \supp(d')$, the routing can be arbitrary for other pairs). The following lemma, the simple proof of which is deferred to \Cref{sec:sublemmamainlemma}, reiterates the sufficient condition for the produced $d'$ and $R'$ to satisfy the requirements of \Cref{lem:samplingorder}.

% \begin{lemma} \label{lem:correctprocess}
\begin{restatable}{lemma}{correctprocesslemma} \label{lem:correctprocess}
Let $\Delta_k = \sum_{p \in \mathcal{P}} w_{k - 1}(p) - w_{k}(p)$. Then,
\begin{itemize}
    \item $d'(s, t) \leq d(s, t)$ for all $(s, t)$,
    \item $\obl(R', d') \leq \gamma$,
    \item $\siz(d') = D - \sum_k \Delta_k$,
\end{itemize}
thus, if $\sum_k \Delta_k \leq \frac{1}{2} D$, $d'$ and $R'$ satisfy the requirements of \Cref{lem:samplingorder}.
\end{restatable}

Thus, it only remains to bound the probability that $\sum_k \Delta_k > \frac{1}{2} D$. We do this by applying a union bound. As there are far too many possible runs of the process or even sequences of values $\delta_k$ to union bound over, we instead apply the union bound over a reduced class of simple objects, \textit{bad patterns}.

\begin{definition}[Bad Patterns]
A \textit{bad pattern} $(b)_k = (b_1, \dots, b_{m})$ is a $m$-tuple of nonnegative integers, such that $\frac{1}{4} D \leq \sum_k \gamma b_k \leq D$.
\end{definition}

% \begin{lemma} \label{lem:badpatternlemma}
\begin{restatable}{lemma}{badpatternlemma} \label{lem:badpatternlemma}
If $\sum_k \Delta_k > \frac{1}{2} D$, there exists a bad pattern $b$, such that $\Delta_k \geq \gamma b_k$ for all $k$.
\end{restatable}

\begin{proof}
Let $b_k = \left\lfloor \frac{\Delta_k}{\gamma} \right\rfloor$. Now,
\begin{itemize}
    \item $b_k$ are integers and $\Delta_k \geq \gamma b_k$ for all $k$,
    \item $\Delta_{k}$ is either $0$ or at least $\gamma$ for every $k$, thus $\gamma \left \lfloor \frac{\Delta_k}{\gamma} \right\rfloor \geq \frac{1}{2} \Delta_k$ and $\sum_k \gamma b_k \geq \frac{1}{2} \sum_k \Delta_k \geq \frac{1}{4} D$,
    \item $\sum_k \gamma b_k \leq \sum_k \Delta_k = D - \siz(d') \leq D$.
\end{itemize}
thus $(b)_k$ is a bad pattern with the desired property.
\end{proof}

\Cref{lem:badpatternlemma} shows that if there is no bad pattern $(b)_k$ such that $\Delta_k \geq \gamma b_k$ for all $k$, we have $\sum_k \Delta_k \leq \frac{1}{2} D$. With \Cref{lem:badcountlemma}, proven by a simple combinatorial computation in \Cref{sec:sublemmamainlemma}, we bound the number of bad patterns. Then, in \Cref{lem:strongproblem}, we prove that for a fixed bad pattern $(b)_i$, the probability that $\Delta_k \geq \gamma b_k$ for all $k$ is small. Afterward, a simple union bound combines these three results to upper bound the probability that $\sum_k \Delta_k > \frac{1}{2} D$, which by \Cref{lem:correctprocess} completes the proof of  \Cref{lem:samplingorder}.

% \begin{lemma} \label{lem:badcountlemma}
\begin{restatable}{lemma}{badcountlemma} \label{lem:badcountlemma}
There are at most $m^{6 D / \alpha}$ bad patterns.
\end{restatable}

\begin{lemma}
\label{lem:strongproblem}
For a bad pattern $(b)_k$,
\begin{equation*}
    \prob{\bigcap_{k \in [|E|]} \Delta_k \geq \gamma b_k} \leq m^{-(h + 7) D / \alpha}.
\end{equation*}
\end{lemma}

\begin{proof}
Fix the bad pattern $(b)_k$ and let $B = \{k \in [|E|] \mid b_k > 0\}$ be the set of overcongesting edges. We have \begin{equation*}
    \prob{\bigcap_{k \in [|E|]} \Delta_k \geq \gamma b_k} \leq \prob{\bigcap_{k \in B} \Delta_k \geq \gamma b_k}.
\end{equation*}

Let $Z_k \subseteq \bigcup_{s, t} \supp(R(s, t))$ be the set of paths that contain edge $e_k$ but no edge $e_{k'}$ for $k' \in B \cap [k - 1]$ and $Z_{k, s, t} = Z_k \cap \supp(R(s, t))$. Now, $Z_{k}$ is the set of paths that could possibly congest edge $e_k$ if the event occurs, as any path crossing an earlier edge has already been cut, as the event occurring implies those edges are overcongested. Let $Y_k$ be a random variable giving the number of sampled paths that cross edge $e_k$ but no earlier bad edge:
\begin{equation*}
    Y_k = \sum_{(s, t) \in \supp(d)} \sum_{p \in Z_{k, s, t}} \sum_i X(s, t)_{i, p}.
\end{equation*}
We show that $\bigcap_{k \in B} Y_k \geq \gamma b_k$ is implied by $\bigcap_{k \in B} \Delta_k \geq \gamma b_k$, and then upper bound the probability of the former event.

To see this, we assume the contrary: that there exists a sampling $x(s, t)_{i, p}$, such that with nonzero probability, $X(s, t)_{i, p} = x(s, t)_{i, p}$ holds for all $s, t, i, p$, and when $X(s, t)_{i, p} = x(s, t)_{i, p}$, $\Delta_{k} \geq \gamma b_{k}$ holds for all $k$, but there is some $k'$ such that $Y_{k'} < \gamma b_{k'}$. But then,

\begin{align*}
    \text{cong}_{k'}    &= \sum_{p \in \calP \mid e_{k'} \in p} w_{k' - 1}(p)\\
                        &\leq \sum_{p \in \calP \mid e_{k'} \in p} w_{0}(p) \prod_{k \in B \cap [k' - 1]} \ind{e_k \not\in p}\\
                        &= \sum_{p \in Z_{k'}} w_{0}(p),
\end{align*}
as for $k \in B$, $w_{k}(p) = w_{k - 1}(p) \ind{e_k \not\in p}$, as $\Delta_k > \gamma b_k > 0$ implies the edge $e_k$ is overcongested, thus all paths crossing it were cut during step $k$, and $w_{k}(p) \leq w_{k - 1}(p)$ holds for all $k$, thus also for $k \not\in B$. Finally,
\begin{align*}
    \sum_{p \in Z_{k'}} w_0(p)
    &= \sum_{(s, t)} \sum_{p \in Z_{k', s, t}} w_0(p)\\
    &= \sum_{(s, t)} \sum_{p \in Z_{k', s, t}} \frac{d(s, t)}{\cou(s, t)} \sum_i x(s, t)_{i, p}\\
    &= \sum_{(s, t) \in \supp(d)} \sum_{p \in Z_{k', s, t}} \sum_i x(s, t)_{i, p}\\
    &= Y_{k'}
\end{align*}
where we use that the demand is special. Now, we have a contradiction: $\obl_{k'} \leq Y_{k'}$, but also $Y_{k'} < \gamma b_{k'} \leq \Delta_k = \obl_{k'}$. Thus, it suffices to bound the probability of the event $\bigcap_{k \in B} Y_k \geq \gamma b_k$.

Next, notice that by definition, for $k \in B$, the sets $Z_k$ are disjoint. Thus, since the random variables $X(s, t)_{i, p}$ are negatively associated, by \Cref{lem:negativeassociationindep}, we have $\prob{\bigcap_{k \in B} Y_k \geq \gamma b_k} \leq \prod_{k \in B} \prob{Y_k \geq \gamma b_k}$.

We now bound the probability $\prob{Y_k \geq \gamma b_k}$. We have
\begin{align*}
    \mathbb{E}\left[Y_k\right]
            &= \mathbb{E}\left[\sum_{(s, t) \in \supp(d)} \sum_{p \in Z_{k, s, t}} \sum_i X(s, t)_{i, p}\right]\\
            &= \sum_{(s, t)} \sum_{p \in Z_{k, s, t}} \sum_i \ind{(s, t) \in \supp(d)} \mathbb{E}\left[X(s, t)_{i, p}\right]\\
            &= \sum_{(s, t)} \sum_{p \in Z_{k, s, t}} \sum_i \frac{d(s, t)}{\cou(s, t)} \prob{R(s,t) = p}\\
            &= \sum_{(s, t)} \sum_{p \in Z_{k, s, t}} d(s, t) \prob{R(s,t) = p}\\
            &\leq \sum_{(s, t)} \sum_{\substack{p \in \supp(R(s, t))\\e_k \in p}} d(s, t) \prob{R(s,t) = p}\\
            &\leq \obl(R, d).
\end{align*}

Let $\delta = \frac{\gamma}{\expec{Y_k}}$. By the above, $\delta \geq \alpha + m^{16(h + 7) \alpha^{-1}} \geq 2$. Thus, as $Y_k$ is a sum of disjoint negatively associated $\{0, 1\}$-random variables, by Chernoff for negatively associated random variables (\Cref{lem:firstchernoffnegativeassociation}),
\begin{equation*}
    \prob{Y_k \geq \gamma b_k}
    = \prob{Y_k \geq \delta b_k \expec{Y_k}}
    \leq \exp\left(\frac{-\expec{Y_k} \delta b_k \ln \left(\delta b_k\right)}{4}\right).
\end{equation*}
Since $\expec{Y_k} \delta = \gamma$ and $\delta \geq m^{16(h + 7) \alpha^{-1}}$, we have $\ln(\delta b_k) \geq \ln(\delta) \geq 16(h + 7) \alpha^{-1} \ln m$, thus,
\begin{equation*}
    \prob{Y_k \geq \gamma b_k}
    \leq \exp\left(-\gamma b_k \frac{16(h + 7)}{4 \alpha} \ln m\right)
    = m^{-\gamma b_k \frac{4(h + 7)}{\alpha}}.
\end{equation*}
Now, with an upper bound on $\prob{Y_k \geq \gamma b_k}$, we can upper bound the product of these terms. Since $b$ is a bad pattern, we have $\sum_{k \in [|E|]} \gamma b_i = \sum_{k \in B} \gamma b_i \geq \frac{1}{4} D$. Thus, as desired,
\begin{align*}
    \prob{\bigcap_{k \in [|E|]} \Delta_k \geq \gamma b_k}
    &\leq \prod_{k \in B} \prob{Y_k \geq \gamma b_k}\\
    &\leq m^{-\left(\sum_{k \in B} \gamma b_k\right) \frac{4(h + 7)}{\alpha}}\\
    &\leq m^{-\frac{(h + 7) D}{\alpha}}. \qedhere
\end{align*}
\end{proof}

With the proof of \Cref{lem:strongproblem} complete, we can now finish the proof of \Cref{lem:samplingorder}. Let $\mathcal{B}$ be the set of bad patterns, and let $E_b$ be the event that, for bad pattern $b \in B$, we have $\bigcap_{k \in [|E|]} \Delta_k \geq \gamma b_k$.

By \Cref{lem:badpatternlemma}, if $\sum_k \Delta_k > \frac{1}{2} D$, there exists a bad pattern $b$ such that $E_b$ happens. Thus, by a union bound,
\begin{equation*}
\prob{\sum_k \Delta_k \leq \frac{1}{2} D}
    \geq 1 - \prob{\bigcup_{b \in \mathcal{B}} E_b}
    \geq 1 - \sum_{b \in \mathcal{B}} \prob{E_b}.
\end{equation*}
Now, by \Cref{lem:badcountlemma} we have $|\mathcal{B}| \leq m^{4D / \alpha}$, and by \Cref{lem:strongproblem} we have $\prob{E_b} \leq m^{-(h + 7)D / \alpha}$. Thus,
\begin{equation*}
    \sum_{b \in \mathcal{B}} \prob{E_b}
    \leq \sum_{b \in \mathcal{B}} m^{-(h + 7)D / \alpha}
    \leq m^{4D / \alpha} m^{-(h + 7)D / \alpha}
    = m^{-(h + 3) D / \alpha}.
\end{equation*}
Thus, by \Cref{lem:correctprocess}, $\calP$ has the desired property with failure probability $m^{-(h+3)D/\alpha}$. Finally,
\begin{equation*}
D = \sum_{s, t} d(s, t) = \sum_{(s, t) \in \supp(d)} \cou(s, t) \geq \alpha |\supp(d)|,
\end{equation*}
thus $m^{-(h+3)D/\alpha} \leq m^{-(h+3)|\supp(d)|}$, as desired. This completes the proof of \Cref{lem:samplingorder}.

\subsection{Finishing the reduction}\label{sec:finishing-the-reduction}
In this section, we sketch the proofs of the union bound over special demands (\Cref{cor:weakspecialdemands}), the reduction from weak to to general routing (\Cref{lem:weaktostrong}), and the reduction from special demands to general demands (\Cref{lem:specialtogeneral}). The full proofs are left to \Cref{sec:fullproofsfinishingreduction}.

To prove the reductions, we need three simple lemmas. The first bounds congestion when routing the sum of two demands, the second gives trivial weak bounds on the congestion of any routing, and the third shows that competitiveness on polynomially-bounded demands implies competitiveness in general.

\begin{restatable}[demand-sum lemma]{lemma}{demandsumlemma}\label{lem:demandsumlemma}
Let $\calP$ be a semi-oblivious routing, $d_1, d_2$ be two demands and $d = d_1 + d_2$ be their sum, and $R_1, R_2$ be two routings on $\calP$. Then, there exists a routing $R$ on $\calP$, such that $\obl(R, d) \leq \obl(R_1, d_1) + \obl(R_2, d_2)$. If $R_1 = R_2$, $R = R_1$ satisfies the inequality.
\end{restatable}

\begin{proof}[Proof sketch (full proof in \Cref{sec:fullproofsfinishingreduction})]
Let $\prob{R(s,t) = p}$ be the linear combination of $\prob{R_1(s,t) = p}$ and $\prob{R_2(s,t) = p}$ weighted by $d_1(s, t)$ and $d_2(s, t)$. Then, every path has weight equal to the sum of its weights in the routings on $d_1$ with $R_1$ and $d_2$ with $R_2$.
\end{proof}

\begin{restatable}[bounded-congestion lemma]{lemma}{boundconglemma}\label{lem:boundcomplemma}
Let $G$ be a graph, $R$ be a routing and $d$ be a demand. Then,
\begin{equation*}
    \frac{\siz(d)}{|E|} \leq \obl(R, d) \leq \siz(d).
\end{equation*}
\end{restatable}
\begin{proof}[Proof sketch (full proof in \Cref{sec:fullproofsfinishingreduction})]
All paths have length at least $1$, thus the average congestion of edges in any routing is at least $\siz(d) / |E|$. Since we require paths to be simple, the congestion of any edge is at most $\siz(d)$.
\end{proof}

\begin{restatable}[poly-sufficiency lemma]{lemma}{capdemandlemma}\label{lem:capdemands}
Let $G$ be a $n$-vertex $m$-edge graph and $R$ be an oblivious routing. For real $r \geq 1$, let $D_{r}$ be the set of demands, where $d(s, t) \in \{0\} \cup [1, r]$. Let $\calP$ be a semi-oblivious routing on $G$ that is $C$-competitive with $R$ on $D_{n^2 m}$. Then, $\calP$ is $2C$-competitive with $R$. 
\end{restatable}
\begin{proof}[Proof sketch (full proof in \Cref{sec:fullproofsfinishingreduction})]
Take any demand, scale it down, and break it into two parts $d_1$ and $d_2$, such that $d_1(s, t) \leq 1$ and $\siz(d_2) = n^2 m$. Any routing has low congestion on $d_1$, and by $C$-competitiveness on $D_{n^2 m}$ there is a routing on $\calP$ with low congestion on $d_2$. Use \Cref{lem:demandsumlemma} and the linearity of congestion to get a good routing on $\calP$ for the original demand. 
\end{proof}

The proof of \Cref{cor:weakspecialdemands} is achieved by a simple union bound over special demands.

\weakspecialdemandCorollary*

\begin{proof}[Proof sketch (full proof in \Cref{sec:fullproofsfinishingreduction})]
There are at most $n^{2k}$ special demands $d$ with $|\supp(d)| = k$. We take a union bound over all special demands, using \Cref{lem:samplingorder} to bound the probability we aren't sufficiently competitive on individual demands. We get a bound of $\sum_{k \geq 1} m^{-(h + 1)k}$ on the probability that a special demand we aren't competitive enough on exists, but this is at most $m^{-h}$, as desired.
\end{proof}

The reduction from weak routing to general routing is in principle proven by simply repeatedly routing half of the demand, but care needs to be taken as at this point we have not proven weak routing for arbitrary demands. To work with the limited set of demands available, we route in full the demand between vertex pairs between which the weak routing routes at least a fourth of the demand, and route no demand between other vertex pairs. This routes at least a third of the original demand and leaves a subdemand that at every vertex pair either equals the original demand or is zero. 

\weaktostrongLemma*

\begin{proof}[Proof sketch (full proof in \Cref{sec:fullproofsfinishingreduction})]
If for every demand $d \in D$, the set $D$ contained every demand $d' \leq d$, we could repeatedly use the weak-competitiveness property of $\calP$ to get a routing for a part of the current demand, until the remaining demand has size at most $\siz(d) / m$. The remaining part can be routed arbitrarily, and \Cref{lem:demandsumlemma} can be used to combine the $\bigO{\log m}$ different routings into a routing for $d$ with $\bigO{C \log m}$-competitive congestion. To get around the weaker condition on $D$, when we apply weak-competitiveness to a demand $d$ to obtain $d'$, we let $d''(s, t) = \ind{d'(s, t) \geq \frac{1}{4} d(s, t)} d(s, t)$. Now, $d - d''$ is in $D$, $\siz(d'') \geq \frac{1}{3} \siz(d)$, and $d'' \leq 4d'$, thus the oblivious routing provided by weak-competitiveness routes $d''$ competitively.
\end{proof}

Finally, the reduction from special demands to general demands is done by bucketing pairs of vertices according to the ratio of demand to the number of sampled paths between them. Every bucket allows a range between two powers of two, thus we only need a logarithmic number of buckets in total, and every bucket can be routed as if the ratio between every vertex pair was equal.

\specialtogeneralLemma*

\begin{proof}[Proof sketch (full proof in \Cref{sec:fullproofsfinishingreduction})]
By \Cref{lem:capdemands} it is sufficient to show competitiveness on polynomially bounded demands. Take one such demand, and split it into $\bigO{\log m}$ parts, such that in every part, the values $\frac{d(s, t)}{\alpha + \cut(s, t)}$ are within a factor of $2$ of each other. Then, for each of these parts, take a larger demand where all of these fractions are equal. Scale those down, and use competitiveness on special demands to show competitiveness on the demand. Finally, combine the routings for the parts using \Cref{lem:demandsumlemma}.
\end{proof}

\section{Integral Semi-Oblivious Routing}\label{sec:integral-oblivious}

In this section, we adapt the results of section \Cref{sec:sparse-semi-oblivious-routing} to the case where we consider integral routings and \zeroonedemands. In this case, the additive $\cut$-term in the sparsity and one logarithmic multiplier in the competitiveness can be avoided.

The definition of an integral semi-oblivious routing $\mathcal{P}$ is like that of a real-valued semi-oblivious routing, except that we define the congestion as the minimum congestion of an integral routing on $\mathcal{P}$, and define competitiveness as a comparison against optimal integral routings.

\begin{definition}[Integral semi-oblivious routing]

An \textit{integral semi-oblivious routing} is a path system $\calP$. For a demand $d$, we define the congestion of integrally routing $d$ via $\calP$:
\begin{equation*}
    \oblint(\calP, d) = \min_{\substack{\text{$R$ a routing on $\calP$}\\\text{$R$ integral on $d$}}} \obl(R, d).
\end{equation*}
We say $\calP$ is \textit{$C$-competitive on a set of demands $D$} for $C \geq 1$ if for every demand $d \in D$ we have $\oblint(\calP, d) \leq C \optint(d)$ and that $\calP$ is \textit{$C$-competitive} for $C \geq 1$ if it is $C$-competitive on the set of all integral demands.
\end{definition}

To prove results for integral semi-oblivious routings, we need the $\alpha$-sparse version of the $(\alpha + \cut)$-sample theorem (\Cref{thm:mainthm}). The statement of this version involves additive terms, which disappear when switching to integral semi-oblivious routings, as the congestion of an integral optimal routing is at least one.

\begin{restatable}[$\alpha$-sample corollary]{corollary}{actualmaincorollary}\label{cor:actualmaincorollary}
Let $G = (V, E)$ be a $n$-vertex graph with polynomially many edges, $R$ be an oblivious routing and $\alpha \in [n]$ be a parameter. Let $\calP$ be an $\alpha$-sample of $R$. Then, with high probability, for every demand $d$,
\begin{equation*}
    \oblreal(\calP, d) \leq \bigO{\log^2(n) \left(\alpha + n^{\bigO{\alpha^{-1}}}\right)} \left(\obl(R, d) + \max_{s, t} d(s, t)\right),
\end{equation*}
and for every \zeroonedemand $d$,
\begin{equation*}
    \oblreal(\calP, d) \leq \bigO{\log(n) \left(\alpha + n^{\bigO{\alpha^{-1}}}\right)} \left(\obl(R, d) + 1\right).
\end{equation*}
\end{restatable}

\begin{proof}[Proof sketch (full proof in \Cref{sec:alphasampleroundinglemma})]
We may assume that $\alpha \geq 2$, as no guarantees for $\alpha = 1$ are made. Create an auxiliary graph $G_2$ with $n + 2n^2$ vertices: the $n$ of the original graph, and $2$ auxiliary vertices for every vertex pair. Connect the vertices for pair $(s, t)$ to $s$ and $t$ respectively with a single edge. Now, the min-cut between any two auxiliary vertices is $1$.

Take an extension $R_2$ of $R$ to an oblivious routing on $G_2$ by routing between the auxiliary vertices for $(s, t)$ by using $R$ to route between $s$ and $t$ and prepending and appending the two bridges, then apply Theorem \ref{thm:mainthm} to $R_2$, $G_2$ and $\alpha - 1$ to obtain a $(\alpha - 1 + \cut)$-sample $\calP_2$ of $R_2$. Map this sample into a semi-oblivious routing $\calP$ on $G$ by having the paths in $\calP$ between $s$ and $t$ correspond to those between the auxiliary vertices of $s$ and $t$ in $\calP_2$.
\end{proof}

Theorems \ref{thm:mainthm} and \ref{cor:actualmaincorollary} prove results for semi-oblivious routings, not integral semi-oblivious routings. To prove the claimed corollaries, we use the following lemma that shows that, for any routing $R$, for an integral demand $d$, there exists a routing $R'$ on $\supp(R)$ that is integral on $d$ and has congestion only a constant multiplicative factor and a logarithmic additive factor higher than the oblivious routing $R$.

\begin{restatable}[Rounding lemma]{lemma}{roundlemma}\label{lem:roundlemma}

    Let $G$ be a $m$-edge graph, $R$ a routing and $d$ a demand. Then, there exists a routing $R'$ on $\supp(R)$ that is integral on $d$, such that
    \begin{equation*}
        \obl(R', d) \leq 2 \obl(R, d) + 3 \ln m.
    \end{equation*}
\end{restatable}

\begin{proof}[Proof sketch (full proof in \Cref{sec:alphasampleroundinglemma})]
Sample $d(s, t)$ paths from $R(s, t)$ and select $R'$ to assign the weight of a path to be equal to the number of times it was sampled. This $R'$ is on $\supp(R)$ and is integral on $d$. Take a union bound over edges, bounding the probability an individual edge over-congests with a Chernoff bound. With nonzero probability, $\obl(R', d)$ satisfies the desired bound, thus a $R'$ on $\supp(R)$ that is integral on $d$ exists.
\end{proof}

The following corollary of Lemma \ref{lem:roundlemma} allows conveniently proving results in the integral domain.
\begin{corollary}\label{cor:integraliseasy}
  Let $G$ be a $m$-edge graph and $\calP$ a path system on $G$. Then, for every integral demand $d$,
\begin{equation*}
    \oblint(\calP, d) \leq 2 \oblreal(\calP, d) + 3 \ln m
\end{equation*}
\end{corollary}
\begin{proof}
Let $R$ be a routing on $\calP$ such that $\oblreal(\calP, d) = \obl(R, d)$. By Lemma \ref{lem:roundlemma}, there exists a routing $R'$ on $\supp(R)$ (thus on $\calP$) that is integral on $d$ such that $\obl(R', d) \leq 2 \obl(R, d) + 3 \ln m$. But now,
\begin{equation*}
    \oblint(\calP, d) \leq \obl(R', d) \leq 2 \obl(R, d) + 3 \ln m = 2 \oblreal(\calP, d) + 3 \ln m.
\end{equation*}
\end{proof}

Using \Cref{thm:mainthm}, \Cref{cor:actualmaincorollary} and \Cref{cor:integraliseasy} on the oblivious routing of \cite{Racke08}, we can prove the statements stated in \Cref{sec:formal-results}:

\begin{theorem}[\cite{Racke08}]\label{thm:logobliviousrouting}
Every $n$-vertex graph has a $\bigO{\log n}$-competitive oblivious routing.
\end{theorem}

With this we recover the result for the logaritmic and low sparsity cases.

\basiccorollary*

\begin{proof}[Proof sketch (full proof in \Cref{sec:fullproofsformalresults})]
Apply \Cref{cor:actualmaincorollary} with $\alpha = \bigO{\frac{\log n}{\log \log n}}$ and the $\bigO{\log n}$-competitive oblivious routing $R$ from \cite{Racke08}, then use \Cref{cor:integraliseasy} to make the semi-oblivious routing integral.
\end{proof}

\basicalphacorollary*

\begin{proof}[Proof sketch (full proof in \Cref{sec:fullproofsformalresults})]
Apply Theorem \ref{cor:actualmaincorollary} with $\alpha$ and the $\bigO{\log n}$-competitive oblivious routing $R$ from \cite{Racke08}, then use \Cref{cor:integraliseasy} to make the semi-oblivious routing integral. The $n^{\bigO{\alpha^{-1}}}$-term covers all logarithmic terms.
\end{proof}

Finally, we recover the result on general integral demands.

\cutcorollary*

\begin{proof}[Proof sketch (full proof in \Cref{sec:fullproofsformalresults})]
Apply \Cref{thm:mainthm} with $\alpha = \bigO{\frac{\log n}{\log \log n}}$ and the $\bigO{\log n}$-competitive oblivious routing $R$ from \cite{Racke08}, then use \Cref{cor:integraliseasy} to make the semi-oblivious routing integral.
\end{proof}

\section{Semi-Oblivious Routing for Completion Time}\label{sec:hopconstrainedsemiobl}

For the results on competitiveness in congestion and dilation, we need to apply our theorem to a \textit{hop-constrained oblivious routing}:

\textbf{Hop-constrained oblivious routing.} We define the optimal integral $h$-hop congestion $\optint^{(h)}(d)$ as the minimum congestion over all routings $R$ that are integral on $d$ with dilation at most $h$.

Let $G$ be a graph and $R$ be an oblivious routing on $G$. We say $R$ is a $h$-hop oblivious routing for $h \geq 1$ with hop-stretch $\beta \geq 1$ and congestion approximation $C \geq 1$ if for all demands $d$ we have $\dil(R, d) \leq \beta h$ and $\obl(R, d) \leq C \opt^{(h)}(d)$.

\begin{theorem}[Theorem 3.1 of \cite{GhaffariHZ21}]\label{thm:hopconstrainedobl}
For every graph $G$ and every $h \geq 1$, there exists a $h$-hop oblivious routing with hop stretch $\bigO{\log^7 n}$ and congestion approximation $\bigO{\log^2(n) \log^2(h \log n)}$.
\end{theorem}

Completion-time competitive semi-oblivious routings are defined analogously to the regular variants:

\begin{definition}[Completion-time competitive integral semi-oblivious routing]

We say a integral semi-oblivious routing $\calP$ is \textit{$C$-completion time competitive on a set of demands $D$} for $C \geq 1$ if for every demand $d \in D$ and routing $R$, there exists a routing $R'$ on $\calP$ such that
\begin{equation*}
    \obl(R', d) + \dil(R', d) \leq C(\obl(R, d) + \dil(R, d)).
\end{equation*}
We say $\calP$ is \textit{$C$-completion time competitive} if it is $C$-completion time competitive on the set of all demands.
\end{definition}

We in fact prove something stronger: the semi-oblivious routing is competitive in \emph{both} congestion and dilation with any demand. This of course implies competitiveness in completion time.

\congdilationcorollary*

\begin{proof}[Proof sketch (full proof in \Cref{sec:fullproofsformalresults})]
Let $h_1 = 1$ and $h_i = \lceil h_{i - 1} \log n \rceil$. For $i \in \left[\left\lceil \frac{\log n}{\log \log n}\right\rceil\right]$, let $R_i$ be a $h_i$-hop oblivious routing from \cite{GhaffariHZ21}. For $\alpha = \bigO{\frac{\log n}{\log \log n}}$, for every $i$, let $\calP_i$ be an $\alpha$-sparse semi-oblivious routing obtained by using \Cref{cor:actualmaincorollary} to $\alpha$ and $R_i$. Define $\calP$ as $P(s, t) := \bigcup_i P_i(s, t)$. This $\calP$ has the desired sparsity.

Now, for a \zeroonedemand $d$ and routing $R$ integral on $d$ with dilation between $h_{j - 1}$ and $h_{j}$, there exists a routing $R'$ on $\calP_{j}$ (thus also on $\calP$) that has congestion $\obl(R, d) \poly \log n$ and dilation $h_j \poly \log n \leq \dil(R, d) \poly \log n $. Finally, we use \Cref{lem:roundlemma} to make the routing integral on $d$.
\end{proof}

The proof of the low-sparsity result \Cref{lem:congdilationsparse} is similar, just scaling the $h_i$s by $n^{\bigO{\alpha^{-1}}}$ instead of $\log n$, in which case only $\bigO{\alpha}$ different scales are needed. The proof is deferred to \Cref{sec:fullproofsformalresults}.

\congdilationsparsecorollary*

\section{Lower bound}\label{sec:lower-bound}

In this section, we present an explicit construction of a family of simple graphs that gives a lower bound on the best achievable competitiveness of sparse semi-oblivious routing in both the fractional and integral settings, even when restricted to permutation demands. As the demands that cause the semi-oblivious routings to exhibit bad behaviour only have demand between vertices between which the minimum cut has size $1$, the lower bounds apply to both $\alpha$-sparsity and $(\alpha + \cut)$-sparsity. 

We construct the class of graphs in two parts: first, we construct an $\bigO{n}$-vertex graph based on $n$ and $\alpha$, where any $(\alpha - 1 + \cut)$-sparse integral semi-oblivious routing can at best be $\alpha^{-1} \lfloor n^{\frac{1}{2 \alpha}} \rfloor$-competitive on permutation demands. Then, we use bridges to connect multiple copies of that graph built with different $\alpha$ from $1$ to $\log n$, creating a $\bigO{n \log n}$-vertex graph where any $(\alpha - 1 + \cut)$-sparse integral semi-oblivious routing can at best be $\alpha^{-1} \lfloor n^{\frac{1}{2 \alpha}} \rfloor$-competitive on permutation demands.

\begin{restatable}{lemma}{fixedalphalowerbound}\label{lem:fixedalphalowerbound}
Let $C(n, k)$ be the $(2n + 2 + k)$-vertex $(2n + 2k)$-edge graph consisting of two $n + 1$-vertex stars and $k$ vertices connected to the two centers of the stars.

Fix $n$ and $\alpha \in [n]$, and let $k = \lfloor n^{\frac{1}{2 \alpha}} \rfloor$. Then, for every $(\alpha - 1 + \text{cut})$-sparse semi-oblivious routing $\calP$ on $C(n, k)$, there exists a permutation demand $d$, such that
\begin{equation*}
    \oblreal(\calP, d) \geq \alpha^{-1} k \cdot \text{opt}_{C(n, k), \mathbb{Z}}(d).
\end{equation*}
\end{restatable}

\begin{figure}[h]
   \centering
   \includegraphics[width=.45\textwidth]{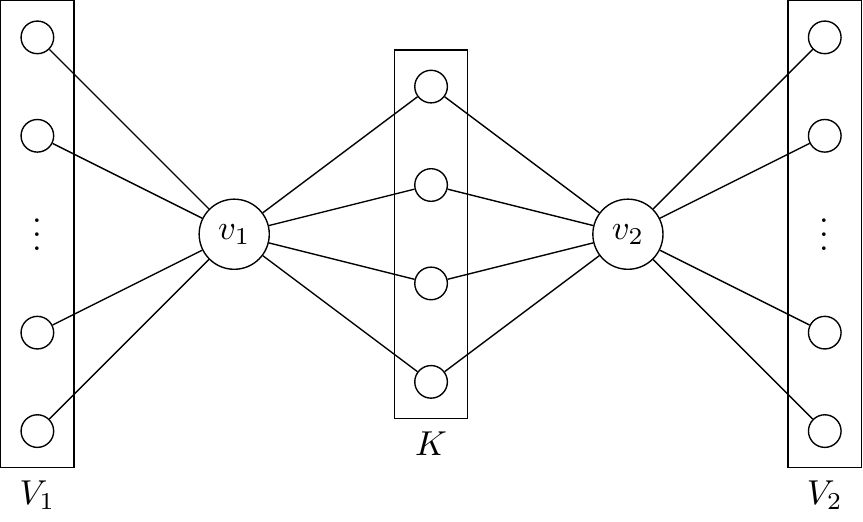}
 \caption{$C(n, k)$ for $n = 256$, $k = 4$. The vertices are labeled as in the proof of \Cref{lem:fixedalphalowerbound}, $|V_1| = |V_2| = 256$ and $|K| = 4$. By \Cref{lem:fixedalphalowerbound}, on this graph, every $2$-sparse semi-oblivious routing is at best $2$-competitive.}
 \label{fig:lowerbound}
\end{figure}

Note that while the semi-oblivious routing doesn't have to be integral, the lower bound uses $\text{opt}_{C(n, k), \mathbb{Z}}$ instead of $\text{opt}_{C(n, k), \mathbb{R}}$. Thus, the lower bound applies to integral semi-oblivious routings as well.

\Cref{lem:fixedalphalowerbound} is not sufficient to prove \Cref{lem:basiclowerbound}, as $C(n, k)$ depends on $\alpha$. \Cref{lem:explicitlowerbound} fixes this by constructing a graph out of multiple $C(n, k)$ built with different $\alpha$. Note that connecting the graphs with bridges does not affect the cuts or routings inside a $C(n, k)$-subgraph.

\begin{restatable}{lemma}{explicitlowerbound}\label{lem:explicitlowerbound}
Let $C(n, k)$ be the $(2n + 2 + k)$-vertex $(2n + 2k)$-edge graph consisting of two $n + 1$-vertex stars and $k$ vertices connected to the two centers of the stars. Let $G(n)$ be a graph built as follows: we make a copy of $C(n, \lfloor n^{\frac{1}{2 \alpha}} \rfloor)$ for every $\alpha \in [\lfloor \log n \rfloor]$, then arbitrarily connect the copies with $\lfloor \log n \rfloor - 1$ edges to make the graph connected.

Fix $n \geq 2$. Then, for every $\alpha \in [n]$ and every $(\alpha - 1 + \text{cut})$-sparse semi-oblivious routing $\calP$ on $G$, there exists a permutation demand $d$, such that
\begin{equation*}
    \oblreal(\calP, d) \geq \alpha^{-1} \lfloor n^{\frac{1}{2 \alpha}} \rfloor \text{opt}_{G(n), \mathbb{Z}}(d).
\end{equation*}
\end{restatable}

The following corollary is a special case of \Cref{lem:explicitlowerbound} for $\alpha$-sparsity, integral semi-oblivious routing and permutation demands.

\begin{corollary}\label{cor:easylowerbound}
Define $G(n)$ as in \Cref{lem:explicitlowerbound}. Then, for every $\alpha \in [n]$, there exists no $\alpha$-sparse integral semi-oblivious routing $\calP$ on $G$ that is $\left(\frac{1}{2} n^{\frac{1}{2\alpha}} \log^{-1} n\right)$-competitive on permutation demands.
\end{corollary}

\begin{proof}
Fix any $\alpha \in [n]$ and a $\alpha$-sparse integral semi-oblivious routing $\calP$. Since every cut has size at least $1$, any $\alpha$-sparse semi-oblivious routing is also $(\alpha + \text{cut} - 1)$-sparse. Thus, by \Cref{lem:explicitlowerbound}, there exists a permutation demand $d$, such that
\begin{equation*}
    \oblint(\calP, d) \geq \oblreal(\calP, d) \geq \lfloor n^{\frac{1}{2 \alpha}} \rfloor \log^{-1}(n) \text{opt}_{G(n), \mathbb{Z}}(d) > \frac{1}{2} n^{\frac{1}{2 \alpha}} \log^{-1}(n) \text{opt}_{G(n), \mathbb{Z}}(d)
\end{equation*}
as if $n^{\frac{1}{2 \alpha}} < 1$, then $\left(\frac{1}{2} n^{\frac{1}{2\alpha}} \log^{-1} n\right) < 1$, and no semi-oblivious routing can be sub-$1$-competitive.
\end{proof}

We leave the proofs of \Cref{lem:basiclowerbound} and \Cref{lem:basicalphalowerbound} to the end of the section. They are both simple applications of \Cref{cor:easylowerbound}.

To prove \Cref{lem:fixedalphalowerbound}, we aim to find a matching $(s_i, t_i)$ between $k$ leaves of the first star and $k$ leaves of the second star, such that there is a set $S'$ of $\alpha$ vertices among the $k$ between the centers of the two stars, such that for every $i$, the $\alpha$ paths between $s_i$ and $t_i$ each contain at least one of the vertices. Then, the demand where $d(s_i, t_i) = 1$ can be routed integrally with congestion $1$, but the semi-oblivious routing cannot route it with congestion less than $\frac{k}{\alpha}$, as every $(s_i, t_i)$-path goes through at least one of the $\alpha$ vertices in $S'$.

To find such a $S'$, we use the pigeonhole principle two times. There are $\binom{k}{\alpha} \leq \sqrt{n}$ different size-$\alpha$ subsets of the vertices between the centers of the two stars, and for every vertex pair $(s, t)$ of a leaf $s$ of the left star and leaf $t$ of the right star, there is a subset $f(s, t)$ of size $\alpha$ such that every path in $P(s, t)$ contains at least one of the vertices. Thus, for every leaf $s$ of the left star, there is a subset $f(s)$ of size $\alpha$ such that there are at least $\sqrt{n}$ leaves $t$ of the right star, such that $f(s, t) = f(s)$. Thus, there is a subset $S'$ of size $\alpha$ such that there exist at least $\sqrt{n}$ leaves $s$ of the first star, such that $f(s) = S'$. By the definition of $f(s)$, the perfect matching must now exist.

\fixedalphalowerbound*

\begin{proof}
If $\alpha \geq k$, the claim is trivial, as no semi-oblivious routing can be sub-$1$-competitive. Thus, we may assume that $\alpha \leq k$.

Let $V_1$ be the set of the $n$ leaves in the first star and $v_1$ be its center, $V_2$ be the set of leaves in the second star and $v_2$ its center, $K$ be the set of the $k$ vertices that are each connected to $v_1$ and $v_2$, and $\mathcal{F}_{\alpha}$ be the family of size-$\alpha$ subsets of $K$.

Let $\calP$ be a $(\alpha - 1 + \text{cut})$-sparse semi-oblivious routing on $C(n, k)$. For $(s, t) \in V_1 \times V_2$, we have $\text{cut}_{C(n, k)}(s, t) = 1$, thus $|P(s, t)| \leq \alpha$. Let $f(s, t) \in \mathcal{F}_\alpha$ be an arbitrary subset of $\alpha$ vertices from $K$, such that every path in $P(s, t)$ contains at least one of the vertices in $f(s, t)$. Such a set is guaranteed to exist since $|P(s, t)| \leq \alpha$ and removing all vertices in $K$ disconnects $V_1$ and $V_2$. For a fixed $s \in V_1$, we have
\begin{equation*}
    n = |V_2| = \sum_{S \in \mathcal{F}_\alpha} |\{t \in V_2 \mid f(s, t) = S\}| \leq |\mathcal{F}_\alpha| \max_{S \in \mathcal{F}_\alpha} |\{t \in V_2 \mid f(s, t) = S\}|.
\end{equation*}
Thus, since $|\mathcal{F}_{\alpha}| = \binom{k}{\alpha} \leq k^{\alpha} \leq n^{\frac{\alpha}{2 \alpha}} = \sqrt{n}$, we have $\max_{S \in \mathcal{F}_{\alpha}} |\{t \in V_2 \mid f(s, t) = S\}| \geq \sqrt{n}$. Thus, there exists $f(s) \in \mathcal{F}_{\alpha}$ such that $|\{t \in V_2 \mid f(s, t) = f(s)\}| \geq \sqrt{n}$. The exact choice of $f(s)$ over sets satisfying the condition can be done arbitrarily. Now, applying the same argument again, we have
\begin{equation*}
    n = |V_1| = \sum_{S \in \mathcal{F}_\alpha} |\{s \in V_1 \mid f(s) = S\}| \leq |\mathcal{F}_\alpha| \max_{S \in \mathcal{F}_\alpha} |\{s \in V_1 \mid f(s) = S\}|.
\end{equation*}
Thus, since $|\mathcal{F}_{\alpha}| \leq \sqrt{n}$, there exists $S' \in \mathcal{F}_{\alpha}$ such that $|\{s \in V_1 \mid f(s) = S'\}| \geq \sqrt{n}$. Let $A \subseteq \{s \in V_1 \mid f(s) = S'\}$ be a subset of this set of size $k \leq \sqrt{n}$. By the choice of $f$, for every $s \in A$ there exist at least $\sqrt{n} \geq |A|$ vertices $t \in V_2$ such that $f(s, t) = S'$. Thus, by Hall's criterion \cite{Hall35} there exists a subset $B \subseteq V_2$ of size $k$ and a perfect matching $(s_i, t_i)$ between $A$ and $B$, such that $f(s_i, t_i) = S'$ for all $i \in [k]$.

Now, define a demand $d$ of size $k$, where $d(s_i, t_i) = 1$ for all $i \in [k]$, and $d(s, t)$ is $0$ for all other pairs $(s, t) \in V \times V$. This demand is a permutation demand, and we have $\optint(d) = 1$, but every path $p \in P(s_i, t_i)$ contains a vertex in $S'$ and thus both of its adjacent edges. Thus, for any routing $R'$ on $\calP'$, the total congestion of edges adjacent to vertices in $S'$ must be at least $2 \siz(d)$, thus $\obl(R', d') \geq \frac{2 \siz(d)}{2 |S'|} = \frac{k}{\alpha}$, as desired. 
\end{proof}

The proof of \Cref{lem:explicitlowerbound} is simple: the bridges we add to connect the graph do nothing to paths between vertices on the same side of the bridge, thus since for every $\alpha$ there is a subgraph, connected to the rest of the graph with bridges, where no good $(\alpha - 1 + \cut)$-sparse semi-oblivious routing exists, no such semi-oblivious routing can exist in the whole graph either.

\explicitlowerbound*

\begin{proof}
Fix $\alpha \in [n]$ and a $(\alpha - 1 + \text{cut})$-sparse semi-oblivious routing $\calP$. If $\alpha^{-1} \lfloor n^{\frac{1}{2 \alpha}} \rfloor \leq 1$, every permutation demand $d$ satisfies the condition. Thus, we may assume that $\alpha < n^{\frac{1}{2 \alpha}}$, thus $\alpha \in [\lfloor \log n \rfloor]$.

Let $G' = (V', E')$ be the $C(n, \lfloor n^{\frac{1}{2 \alpha}} \rfloor)$-subgraph of $G$. By the construction of $G$, all edges out of $G'$ are bridges, thus $\cut(s, t) = \text{cut}_{G'}(s, t)$ and all paths in $P(s, t)$ are contained in $G'$ for all $(s, t) \in V'$. Thus, the restriction of $\calP$ to $V' \times V'$ is $(\alpha - 1 + \text{cut})$-sparse, and by \Cref{lem:fixedalphalowerbound}, there exists a permutation demand $d'$ such that for every routing $R'$ on the restriction of $\calP$ to $V' \times V'$, $\text{obl}_{G'}(R', d') \geq \alpha^{-1} \lfloor n^{\frac{1}{2 \alpha}} \rfloor \text{opt}_{G', \mathbb{Z}}(d')$. Thus, letting $d(s, t) = d'(s, t)$ for $(s, t) \in V' \times V'$ and $d(s, t) = 0$ elsewhere, for every routing $R'$ on $\calP'$, we have $\obl(R', d) \geq \alpha^{-1} \lfloor n^{\frac{1}{2 \alpha}} \rfloor \optint(d)$, as desired.
\end{proof}

Both \Cref{lem:basiclowerbound} and \Cref{lem:basicalphalowerbound} are simple corollaries of \Cref{cor:easylowerbound}, though they involve some asymptotic analysis. In particular, we have $n^{\frac{1}{2\alpha}} = 2^{\frac{\log n}{2\alpha}} = \log^{\frac{\log n}{2 \alpha \log \log n}}(n)$, for $\alpha = o(\frac{\log n}{\log \log n})$ a $n^{\bigO{\alpha^{-1}}}$-term covers any polylogarithmic terms, and by \Cref{cor:easylowerbound} any $\alpha$-sparse integral semi-oblivious routing must be super-poly-logarithmically competitive. 

\basiclowerboundcorollary*

\begin{proof}
Let $\mathcal{G} = \{G(n) : n \in \mathbb{Z}_{\ge 1}\}$ be the family of graphs containing every $G(n)$ as defined in \Cref{lem:explicitlowerbound}. Fix $g_1$ and $g_2$ and let $r(n) = \frac{\log n}{2 g_1(n) \log \log n} = \omega(1)$. Then,
\begin{equation*}
    n^{\frac{1}{2 g_1(n)}} = 2^{\frac{\log n}{2 g_1(n)}} = 2^{r(n) \log \log n} = \log^{r(n)} n.
\end{equation*}
Fix a graph $G(n)$ and a $g_1(n)$-sparse integral semi-oblivious routing $\calP$ on $G(n)$. By \Cref{cor:easylowerbound}, $\calP$ cannot be
\begin{equation*}
    \frac{1}{2} n^{\frac{1}{2 g_1(n)}} \log^{-1}(n) = \frac{1}{2} \log^{r(n) - 1} n = \log^{\omega(1)} n
\end{equation*}
-competitive on $G(n)$ on \zeroonedemands. Let $V(n)$ be the vertex set of $G(n)$. Then, $|V(n)| = \bigO{n \log n}$ and
\begin{equation*}
    g_2(|V(n)|) = \bigO{g_2(n)} = o\left(\log^{\omega(1)} n\right)
\end{equation*}
since $g_2 = \poly \log n$. Thus, for large enough $n_0$, for any $n \geq n_0$, for any $n$-vertex graph $G \in \mathcal{G}$, there exists no $g_1(n)$-sparse integral semi-oblivious routing that is $g_2$-competitive on all permutation demands, as desired.
\end{proof}

\basicalphalowerboundcorollary*

\begin{proof}
Let $\mathcal{G} = \{G(n) : n \in \mathbb{Z}_{\ge 1}\}$ be the family of graphs containing every $G(n)$ as defined in \Cref{lem:explicitlowerbound}. Fix $\epsilon > 0$ and $\alpha = o\left(\frac{\log n}{\log \log n}\right)$ and let $r(n) = \frac{\epsilon \log n}{\alpha(n) \log \log n} = \omega(1)$. Then,
\begin{equation*}
    n^{\left(\frac{1}{2} - \epsilon\right)\alpha(n)^{-1}} = n^{\frac{1}{2 \alpha(n)}} n^{-\epsilon \alpha(n)^{-1}} = n^{\frac{1}{2 \alpha(n)}} 2^{-\frac{\epsilon \log n}{\alpha(n)}} = n^{\frac{1}{2 \alpha(n)}} 2^{-r(n) \log \log n} = n^{\frac{1}{2 \alpha(n)}} \log^{-r(n)} n.
\end{equation*}
Fix a graph $G(n)$ and a $\alpha$-sparse integral semi-oblivious routing $\calP$ on $G(n)$. By \Cref{cor:easylowerbound}, $\calP$ cannot be
\begin{equation*}
     \frac{1}{2} n^{\frac{1}{2 \alpha}} \log^{-1}(n) = n^{\left(\frac{1}{2} - \epsilon\right) \alpha^{-1}} \left(\frac{1}{2} \log^{r(n) - 1} n\right)
\end{equation*}
-competitive on $G(n)$ on \zeroonedemands. Let $V(n)$ be the vertex set of $G(n)$. Then, $|V(n)| = \bigO{n \log n}$ and
\begin{equation*}
    |V(n)|^{\left(\frac{1}{2} - \epsilon\right) \alpha^{-1}} = n^{\left(\frac{1}{2} - \epsilon\right) \alpha^{-1}} \bigO{\log n}^{\left(\frac{1}{2} - \epsilon\right) \alpha^{-1}} \leq n^{\left(\frac{1}{2} - \epsilon\right) \alpha^{-1}} \bigO{\log n} = o\left(n^{\left(\frac{1}{2} - \epsilon\right) \alpha^{-1}} \left(\frac{1}{2} \log^{r(n) - 1} n\right)\right).
\end{equation*}
Thus, for large enough $n_0$, for any $n$-vertex graph $G \in \mathcal{G}$, there exists no $\alpha$-sparse integral semi-oblivious routing on $G$ that is $\left(n^{\left(\frac{1}{2} - \epsilon\right) \alpha^{-1}}\right)$-competitive on all permutation demands.
\end{proof}

\bibliographystyle{alpha}
\bibliography{refs} 

\appendix

\section{Deferred Proofs}\label{sec:deferred}

\subsection{Sub-Lemmas of the Main Lemma}\label{sec:sublemmamainlemma}

\correctprocesslemma*
\begin{proof} By definition of $f$, we have $w_{k}(p) \leq w_{k - 1}(p)$. Now,
\begin{itemize}
    \item $d'(s, t) \leq d(s, t)$, as
    \begin{equation*}
        d'(s, t) = \sum_{p \in P(s, t)} w_{|E|}(p) \leq \sum_{p \in P(s, t)} w_{0}(p) = d(s, t).
    \end{equation*}
    
    \item Fix an edge $e_k \in E$. The congestion on edge $e_k$ in the routing of $d'$ with $R'$ is
    \begin{align*}
        \sum_{(s, t) \in \supp(d')} d'(s, t) \prob{e_k \in R'(s, t)}
        &= \sum_{(s, t) \in \supp(d')} \sum_{p \in P(s, t)} w_{|E|}(p) \ind{e_k \in p}\\
        &= \sum_{p \in \calP} w_{|E|}(p) \ind{e_k \in p}.
    \end{align*}
    If $\text{cong}_k \leq \gamma$, we have
    \begin{equation*}
        \sum_{p \in \calP} w_{|E|}(p) \ind{e_k \in p} \leq \sum_{p \in \calP} w_{k - 1}(p) \ind{e_k \in p} = \text{cong}_k \leq \gamma,
    \end{equation*}
    otherwise, $w_{k}(p) = w_{k - 1}(p) \ind{e_k \not\in p}$, and
    \begin{equation*}
        \sum_{p \in \calP} w_{|E|}(p) \ind{e_k \in p} \leq \sum_{p \in \calP} w_{k - 1}(p) \ind{e_k \not\in p} \ind{e_k \in p} = 0.
    \end{equation*}
    
    \item $\sum_{s, t} d'(s, t) = D - \sum_k \Delta_k$, as
    \begin{equation*}
        \sum_{s, t} d'(s, t) = \sum_{p \in \calP} w_{|E|}(p) = \sum_{p \in \calP} w_{0}(p) - \sum_{k} \Delta_k = \sum_{s, t} d(s, t) - \sum_k \Delta_k = D - \sum_k \Delta_k
    \end{equation*}
\end{itemize}
\end{proof}

\badcountlemma*
\begin{proof}
First, recall that $d(s, t) \in \{0, \cou(s, t)\}$. 
\begin{itemize}
    \item $\gamma \geq \alpha \obl(R, d) \geq \frac{\alpha d(s, t)}{\cut(s, t)} \geq \alpha$ for $(s, t) \in \supp(d)$, 
    \item $D = \sum_{(s, t)} d(s, t) = \sum_{(s, t) \in \supp(d)} \cou(s, t) \leq n^2 (\alpha + m) \leq 2n^2 m \leq 2m^3$.
\end{itemize}

The number of bad patterns is at most $\sum_{s = 1}^{\lfloor D / \gamma \rfloor} \binom{s + m - 1}{s}$, where $s$ goes over $\sum_k b_k$, and $\binom{s + m - 1}{s}$ is the number of ways to select $b_k$ given $s = \sum_k b_k$. Now,
\begin{equation*}
    \sum_{s = 1}^{\lfloor D / \gamma \rfloor} \binom{s + m - 1}{s} \leq \lfloor D / \gamma \rfloor \binom{\lfloor D / \gamma \rfloor + m}{\lfloor D / \gamma \rfloor} \leq (m + \lfloor D / \gamma \rfloor)^{\lfloor D / \gamma \rfloor} \leq (m + 2m^3)^{\lfloor D / \gamma \rfloor}.
\end{equation*}
Finally, $m + 2m^3 \leq 3m^3 \leq m^4$ and $\lfloor D / \gamma \rfloor \leq \lfloor D / \alpha \rfloor \leq D / \alpha$ as $\gamma \geq \alpha$, thus there are at most $m^{4D / \alpha}$ bad patterns.
\end{proof}

\subsection{Lemmas from \Cref{sec:finishing-the-reduction}}\label{sec:fullproofsfinishingreduction}

\demandsumlemma*

\begin{proof}
  For $(s, t) \in \supp(d)$, we let
  \begin{align*}
    \prob{R(s,t) = p} := \frac{d_1(s, t) \prob{R_1(s,t) = p} + d_2(s, t) \prob{R_2(s,t) = p}}{d(s, t)} && \text{for all $p$.}
  \end{align*}
  For other $(s, t)$, let $R(s, t) = R_1(s, t)$. If $R_1 = R_2$, we have $R = R_1$. Now,
\begin{align*}
    \obl(R, d)
    &= \max_{e \in E} \sum_{(s, t) \in \supp(d)} d(s, t) \ind{e \in R(s, t)}\\
    &= \max_{e \in E} \sum_{(s, t) \in \supp(d)} \left(d_1(s, t) \ind{e \in R_1(s, t)} + d_2(s, t) \ind{e \in R_2(s, t)}\right)\\
    &\leq \max_{e \in E} \sum_{(s, t) \in \supp(d)} d_1(s, t) \ind{e \in R_1(s, t)} + \max_{e \in E} \sum_{(s, t) \in \supp(d)} d_2(s, t) \ind{e \in R_2(s, t)}\\
    &= \obl(R_1, d_1) + \obl(R_2, d_2).
\end{align*}
as desired.
\end{proof}

\boundconglemma*

\begin{proof}
We have
\begin{align*}
    \obl(R, d)  &= \max_{e \in E} \sum_{(s, t) \in \supp(d)} d(s, t) \ind{e \in R(s, t)}\\
                &\geq \frac{1}{|E|} \sum_{e \in E} \sum_{(s, t) \in \supp(d)} d(s, t) \ind{e \in R(s, t)}\\
                &\geq \frac{1}{|E|} \sum_{(s, t) \in \supp(d)} d(s, t)\\
                &= \frac{\siz(d)}{|E|}.
\end{align*}
For the upper bound,
\begin{equation*}
    \obl(R, d) = \max_{e \in E} \sum_{(s, t) \in \supp(d)} d(s, t) \ind{e \in R(s, t)} \leq \max_{e \in E} \sum_{(s, t) \in \supp(d)} d(s, t) = \siz(d).
\end{equation*}
\end{proof}

\capdemandlemma*

\begin{proof}
Let $d'$ be an arbitrary demand. It suffices to show that $\calP$ is $2C$-competitive with $R$ on $d'$.

Let $d = rd'$ be $d$ scaled by a real $r$ such that $\siz(d) = n^2 m$. Let $d_1(s, t) = d(s, t) \ind{d(s, t) \geq 1}$ and $d_2 = d - d_1$. Since $d_1 \in D_{n^2 m}$, there exists an oblivious routing $R'$ on $\calP$, such that $\obl(R', d_1) \leq C \obl(R, d_1) \leq C \obl(R, d)$. By \Cref{lem:boundcomplemma}, we have
\begin{equation*}
    \obl(R', d_2) \leq \text{siz}(d_2) \leq n^2 = \frac{\siz(d)}{m} \leq \obl(R, d) \leq C \obl(R, d).
\end{equation*}
Now, by \Cref{lem:demandsumlemma},
\begin{equation*}
    \obl(R', d) \leq \obl(R', d_1) + \obl(R', d_2) \leq C \obl(R, d) + C \obl(R, d) = 2C \obl(R, d).
\end{equation*}
Thus, since $\obl(R, \cdot)$ is linear, $\obl(R', d') \leq 2 C \obl(R, d')$, as desired.
\end{proof}

\specialtogeneralLemma*

\begin{proof}[Proof of \Cref{lem:specialtogeneral}]
For real $r$, let $D_{r}$ be the set of demands where $d(s, t) \in \{0\} \cup [1, r]$. Fix $d \in D_{n^2 m}$. We'll show that $\calP$ is $\bigO{C \log m}$-competitive with $R$ on $d$. Thus, by \Cref{lem:capdemands}, $\calP$ is $\bigO{C \log m}$-competitive with $R$.

Let $\cou(s, t) = \alpha + \cut(s, t)$ and $l = \lceil \log (n^2 m) \rceil + 1$. For $i \in [2l]$, let
\begin{align*}
    d_i(s, t) &= d(s, t) \ind{\frac{d(s, t)}{\cou(s, t)} \in [2^{i - l - 1}, 2^{i - l})}\\
    d'_i(s, t) &= \cou(s, t) \ind{(s, t) \in \supp(d_i)}.
\end{align*}
Now, $d = \sum_{i = 1}^{2l} d_i$ and $2^{i - l - 1} d'_i \leq d_i < 2^{i - l} d'_i$. For every $i$, since $d'_i \in \spec$, there exists an oblivious routing $R'_i$ on $\calP$ such that $\obl(R'_i, d'_i) \leq C \obl(R, d'_i)$, thus by the linearity of $\obl$ in $d$,
\begin{equation*}
    2^{l-i} \obl(R'_i, d_i) \leq \obl(R'_i, d'_i) \leq C \obl(R, d'_i) \leq C 2^{l-i+1} \obl(R, d_i),
\end{equation*}
thus $\obl(R'_i, d_i) \leq 2C \obl(R, d_i) \leq 2C \obl(R, d)$. Thus, by \Cref{lem:demandsumlemma}, there exists an oblivious routing $R'$ on $\calP$ such that
\begin{equation*}
    \obl(R', d) \leq \sum_{i = 1}^{2l} \obl(R'_i, d_i) = \sum_{i : \supp(d_i) \neq \emptyset} \obl(R'_i, d_i) \leq 2C |\{i : \supp(d_i) \neq \emptyset\}| \obl(R, d).
\end{equation*}
We have $|\{i : \supp(d_i)  \neq \emptyset\}| \leq 2l$, thus $\obl(R', d) \leq 4Cl \obl(R, d) = \bigO{C \log m} \obl(R, d)$.
\end{proof}

\weaktostrongLemma*

\begin{proof}
Fix a demand $d \in D$, and let $d_0 = d$ and $s = [\lceil \log_{3/2} m \rceil]$. For $i \in [s]$, we'll define demands $d_i \in D$ and routings $R_i$ on $\calP$ such that $d_{i} \leq d_{i - 1}$, $\siz(d_i) \leq \left(\frac{2}{3}\right)^{i} \siz(d_0)$ and $\obl(R_i, d_i - d_{i - 1}) \leq 4C \obl(R, d)$. Then, by \Cref{lem:demandsumlemma}, there exists a routing $R'$ on $\calP$ such that $\obl(R', d_0 - d_s) \leq (4Cs) \obl(R, d)$. Finally, by \Cref{lem:boundcomplemma}, $\obl(R', d_s) \leq \siz(d) / m \leq \obl(R, d)$, thus $\obl(R', d) \leq (4Cs + 1) \obl(R, d) \leq \bigO{C \log m} \obl(R, d)$.

Now, fix $i \in [s]$. by the $C$-weak-competitiveness of $\calP$ with $R$ on $D$, there exists a demand $d'$ and a routing $R_i$ such that $d' \leq d_{i - 1}$, $\siz(d') \geq \frac{1}{2} \siz(d_{i - 1})$ and $\obl(R_i, d') \leq C \obl(R, d_{i - 1}) \leq C \obl(R, d)$. Let $d_i(s, t) = \ind{d'(s, t) < \frac{1}{4} d_{i - 1}(s, t)} d_{i - 1}(s, t)$. Now,
\begin{align*}
    \siz(d_i)   &= \sum_{s, t} \ind{d'(s, t) < \frac{1}{4} d_{i - 1}(s, t)} d_{i - 1}(s, t)\\
                &= \sum_{s, t} \frac{4}{3} \ind{d'(s, t) < \frac{1}{4} d_{i - 1}(s, t)} (d_{i - 1}(s, t) - d'(s, t))\\
                &\leq \sum_{s, t} \frac{4}{3} (d_{i - 1}(s, t) - d'(s, t))\\
                &= \frac{4}{3}(\siz(d_{i - 1}) - \siz(d'))\\
                &\leq \frac{2}{3} \siz(d_{i - 1})
\end{align*}
and
\begin{equation*}
    d_i(s, t) - d_{i - 1}(s, t) = d_i(s, t) \ind{d'(s, t) \geq \frac{1}{4} d_{i - 1}(s, t)} \leq 4 d'(s, t),
\end{equation*}
thus $\obl(R_i, d_i - d_{i - 1}) \leq \obl(R_i, 4d') \leq 4C\obl(R, d)$, as desired.
\end{proof}

\weakspecialdemandCorollary*

\begin{proof}[Proof of \Cref{cor:weakspecialdemands}]
Let $\calP$ be an $\alpha$-sample of $R$. For a fixed special demand $d \in \spec$, let $E_{d}$ be the event that $\calP$ is $\left(\alpha + m^{16(h + 7) \alpha^{-1}}\right)$-weakly competitive with $R$ on $d$. Let $E = \bigcap_{d \in \spec} E_d$ be the event that $\calP$ is $\left(\alpha + m^{16(h + 7) \alpha^{-1}}\right)$-weakly competitive with $R$ on $\spec$. By \Cref{lem:samplingorder}, $\prob{\overline{E_b}} = 1 - \prob{E_b} \leq m^{-(h + 3) |\supp(d)|}$, thus by a union bound,
\begin{equation*}
    \prob{\overline{E}} = \prob{\bigcup_{d \in \spec} \overline{E_d}} \leq \sum_{d \in \spec} \prob{\overline{E_b}} = \sum_{k} \sum_{\substack{d \in \spec\\|\supp(d)| = k}} \prob{\overline{E_b}} \leq \sum_{k} \sum_{\substack{d \in \spec\\|\supp(d)| = k}} m^{-(h + 3)k}
\end{equation*}
where $k$ goes from $1$ to $n^2$, as on an empty demand, $\calP$ is $0$-competitive with $R$ with probability $1$. Let $s_k := |\{d \in \spec : |\supp(d)| = k\}|$. We have $s_k \leq \binom{n^2}{k} \leq n^{2k} \leq m^{2k}$, since a demand $d \in \spec$ is uniquely defined by its support and $\supp(d) \subset V \times V$. Thus,
\begin{equation*}
     \sum_{k} \sum_{\substack{d \in \spec\\|\supp(d)| = k}} m^{-(h + 3)k} = \sum_{k} s_k m^{-(h + 3)k} \leq \sum_{k} m^{-(h + 1)k} \leq m^{-h}
\end{equation*}
as $m \geq n \geq 3$ and $h \geq 1$. Thus $\prob{E} = 1 - \prob{\overline{E}} \geq 1 - m^{-h}$, as desired.
\end{proof}

\subsection{$\alpha$-Sample Corollary and the Rounding Lemma}\label{sec:alphasampleroundinglemma}

\actualmaincorollary*

\begin{proof}
We may assume that $\alpha \geq 2$, as no guarantees for $\alpha = 1$ are made. Let $G_2 = (V_2, E_2)$ be the $n + 2n^2$-vertex graph that is $G$ with two added vertices $v_{s, t, 1}$ and $v_{s, t, 2}$ for every vertex pair in the original graph and two added edges $(v_{s, t, 1}, s)$ and $(t, v_{s, t, 2})$ connecting the new vertices to $s$ and $t$ respectively.

Let $g$ be a function mapping paths in $G$ to paths in $G_2$, such that for $p \in P(s, t)$, $g(p)$ is the $(v_{s, t, 1}, v_{s, t, 2})$-path $(v_{s, t, 1}, s) + p + (t, v_{s, t, 2})$ where $+$ for paths denotes concatenation. Let $R_2$ be an arbitrary oblivious routing on $\calP_2$ such that $R_2(v_{s, t, 1}, v_{s, t, 2} = g(R(s, t))$. With slight abuse of notation, for a demand $d$ on $G$, let $d_2$ be a demand on $G_2$ such that $d_2(v_{s, t, 1}, v_{s, t, 2}) = d(s, t)$ and $d_2$ is zero for all other vertex pairs. Now, for every demand $d$ on $G$, we have
\begin{equation*}
    \text{cong}_{G_2}(R_2, d_2) = \max\left(\obl(R, d), \max_{s, t} d(s, t)\right) \leq \left(\obl(R, d) + \max_{s, t} d(s, t)\right).
\end{equation*}

Now, let $D$ be the set of \zeroonedemands on $G$ and $D_2 = \{d_2 : d \in D\}$ be a set of demands on $G_2$. Then, $D_2$ satisfies the condition for the subset of \zeroonedemands that \Cref{thm:mainthm} requires. We apply \Cref{thm:mainthm} to $G_2$, $R_2$ and $\alpha - 1$ to obtain a $(\alpha + \text{cut}_{G_2})$-sample $\calP_2$ such that, with high probability, $\calP_2$ is
\begin{equation*}
    \bigO{\log^2(n) \left(\alpha + n^{\bigO{\alpha^{-1}}}\right)}
\end{equation*}
-competitive with $R_2$ and
\begin{equation*}
    \bigO{\log(n) \left(\alpha + n^{\bigO{\alpha^{-1}}}\right)}
\end{equation*}
-competitive with $R_2$ on $D_2$. From this point on, assume both hold.

Let $\calP$ be the semi-oblivious routing on $G$ such that $P(s, t) = \{g^{-1}(p) : p \in P_2(v_{s, t, 1}, v_{s, t, 2})\}$. The distribution of $\calP$ is identical to that of an $\alpha$-sample of $R$, since $P_2(v_{s, t, 1}, v_{s, t, 2})$ consists of $\alpha - 1 + \text{cut}_{G_2}(v_{s, t, 1}, v_{s, t, 2}) = \alpha$ samples from $R_2(v_{s, t, 1}, v_{s, t, 2}) = g(R(s, t))$. Additionally, for every demand $d$ on $G$, we have
\begin{align*}
    \oblreal(\calP, d) &\leq \oblreal(\calP_2, d_2)\\
                        &\leq \bigO{\log^2(n) \left(\alpha + n^{\bigO{\alpha^{-1}}}\right)} \obl(R_2, d_2)\\
                        &= \bigO{\log^2(n) \left(\alpha + n^{\bigO{\alpha^{-1}}}\right)} (\obl(R, d) + \max_{s, t} d(s, t)),
\end{align*}
and for $d \in D$,
\begin{align*}
    \oblreal(\calP, d) &\leq \oblreal(\calP_2, d_2)\\
                        &\leq \bigO{\log(n) \left(\alpha + n^{\bigO{\alpha^{-1}}}\right)} \obl(R_2, d_2)\\
                        &= \bigO{\log(n) \left(\alpha + n^{\bigO{\alpha^{-1}}}\right)} (\obl(R, d) + 1),
\end{align*}
as desired.
\end{proof}

\roundlemma*

\begin{proof}
For $(s, t) \in V \times V$, $i \in [d(s, t)]$ and $p \in \supp(R(s, t))$, let $X(s, t)_{i, p}$ be 0/1-random variables such that $\sum_i X(s, t)_{i, p} = 1$, and define $R'$ such that $\prob{R'(s,t) = p} := \frac{1}{d(s, t)} \sum_{i} X(s, t)_{i, p}$. Then, $R'$ is a routing on $\supp(R)$ that is integral on $d$. For $e \in E$, let
\begin{equation*}
    Y_{e} = \sum_{(s, t)} \sum_{\substack{p \in \supp(R((s, t))\\e \in p}} \sum_i X(s, t)_{i, p}.
\end{equation*}
Now, for an edge $e$, we have
\begin{equation*}
    \obl(R', d, e) = \sum_{(s, t)} \sum_{\substack{p \in \supp(R((s, t))\\e \in p}} d(s, t) \prob{R'(s,t) = p} = \sum_{(s, t)} \sum_{\substack{p \in \supp(R((s, t))\\e \in p}} \sum_{i} X(s, t)_{i, p} = Y_e.
\end{equation*}
We have $\expec{\sum_{i} X(s, t)_{i, p}} = d(s, t) \prob{R(s,t) = p}$, thus \begin{equation*}
    \obl(R, d, e) = \sum_{(s, t)} \sum_{\substack{p \in \supp(R((s, t))\\e \in p}} d(s, t) \prob{R(s,t) = p} = \expec{\sum_{(s, t)} \sum_{\substack{p \in \supp(R((s, t))\\e \in p}} \sum_{i} X(s, t)_{i, p}} = \expec{Y_e}
\end{equation*}
thus $\expec{Y_e} \leq \obl(R, d)$, and by a union bound,
\begin{align*}
    \prob{\obl(R', d) \geq 2 \obl(R, d) + 3 \ln m}
        &\leq \sum_{e \in E} \prob{Y_e \geq 2 \obl(R, d) + 3\ln m}\\
        &\leq \sum_{e \in E} \prob{Y_e \geq 2\expec{E}\left[Y_e\right] + 3\ln m}.
\end{align*}

By \Cref{lem:sumnegativeassociation} and \Cref{lem:indnegativeassociation}, the variables $X(s, t)_{i, p}$ are negatively associated, thus $Y_e$ is the sum of negatively associated 0/1-random variables. Letting $\delta_e = 1 + \frac{3 \ln m}{\expec{Y_e}} \geq 2$, by Chernoff (\Cref{lem:secondchernoffnegativeassociation}),
\begin{equation*}
    \prob{Y_e \geq 2\expec{Y_e} + 3 \ln m} = \prob{Y_e \geq (1 + \delta_e) \expec{Y_e}} \leq \exp\left(-\frac{\delta_e^2 \expec{Y_e}}{2 + \delta_e}\right) < m^{-1}.
\end{equation*}
Thus,
\begin{equation*}
    \prob{\obl(R', d) \geq 2 \obl(R, d) + 3 \ln m} \leq \sum_{e \in E} \prob{Y_e > \expec{Y_e} + 3\ln m} < \sum_{e \in E} m^{-1} < 1.
\end{equation*}
Thus, $R'$ has the required properties with positive probability, thus a $R'$ satisfying the properties exists.
\end{proof}

\subsection{Results of \Cref{sec:formal-results}}\label{sec:fullproofsformalresults}

\basiccorollary*

\begin{proof}
Let $R$ be a $\bigO{\log n}$-competitive oblivious routing on $G$ that exists by \cite{Racke08} and let $\calP$ be an $\alpha$-sample of $R$ for $\alpha$ we'll determine later. By \Cref{cor:actualmaincorollary}, with high probability, for every \zeroonedemand $d$,
\begin{align*}
    \oblreal(\calP, d)
        &\leq \bigO{\log(n) \left(\alpha + n^{\bigO{\alpha^{-1}}}\right)} \left(\obl(R, d) + 1\right)\\
        &\leq \bigO{\log(n) \left(\alpha + n^{\bigO{\alpha^{-1}}}\right)} \bigO{\log n} \optint(d)\\
        &= \bigO{\log(n)^2 \left(\alpha + n^{\bigO{\alpha^{-1}}}\right)} \optint(d).
\end{align*}
Assume it is. Now, by \Cref{cor:integraliseasy},
\begin{equation*}
    \oblint(\calP, d) \leq 2 \oblreal(\calP, d) + 3\ln m \leq \bigO{\log(n)^2 \left(\alpha + n^{\bigO{\alpha^{-1}}}\right)} \optint(d).
\end{equation*}
Thus, since $n^{\frac{\log \log n}{2 \log n}} = 2^{\frac{1}{2} \log \log n} = \sqrt{\log n}$, there exists $\alpha = \bigO{\frac{\log n}{\log \log n}}$ for which $n^{\bigO{\alpha^{-1}}} = \bigO{\sqrt{\log n}}$. For such $\alpha$, we obtain a $\bigO{\frac{\log n}{\log \log n}}$-sparse integral semi-oblivious routing that is $\bigO{\frac{\log^3 n}{\log \log n}}$-competitive on \zeroonedemands.

\end{proof}

\basicalphacorollary*

\begin{proof}
Let $R$ be a $\bigO{\log n}$-competitive oblivious routing on $G$ that exists by \cite{Racke08}, and let $\calP$ be an $\alpha$-sample of $R$. By \Cref{cor:actualmaincorollary}, with high probability, for every \zeroonedemand $d$,
\begin{align*}
    \oblreal(\calP, d)
        &\leq \bigO{\log(n) \left(\alpha + n^{\bigO{\alpha^{-1}}}\right)} \left(\obl(R, d) + 1\right)\\
        &\leq \bigO{\log(n) \left(\alpha + n^{\bigO{\alpha^{-1}}}\right)} \bigO{\log n} \optint(d)\\
        &= \bigO{\log(n)^2 \left(\alpha + n^{\bigO{\alpha^{-1}}}\right)} \optint(d).
\end{align*}
Assume it is. Now, by \Cref{cor:integraliseasy},
\begin{equation*}
    \oblint(\calP, d) \leq 2 \oblreal(\calP, d) + 3\ln m \leq \bigO{\log(n)^2 \left(\alpha + n^{\bigO{\alpha^{-1}}}\right)} \optint(d).
\end{equation*}

Since $\alpha = o\left(\frac{\log n}{\log \log n}\right)$, we have $n^{\bigO{\alpha^{-1}}} = n^{\omega\left(\frac{\log \log n}{\log n}\right)} = \log^{\omega(1)} n$. Thus, for large enough $n$, we have
\begin{equation*}
    \oblint(\calP, d) \leq \bigO{\log(n)^2 \left(\alpha + n^{\bigO{\alpha^{-1}}}\right)} \optint(d) \leq n^{2 \bigO{\alpha^{-1}}} \optint(d) 
\end{equation*}
which is still $n^{\bigO{\alpha^{-1}}} \optint(d)$. Since we can select the function $\bigO{\alpha^{-1}}$ freely, we can make the claim hold for small $n$ as well.
\end{proof}

\cutcorollary*

\begin{proof}
Let $R$ be a $\bigO{\log n}$-competitive oblivious routing on $G$ that exists by \cite{Racke08} and let $\calP$ be an $(\alpha + \cut)$-sample of $R$ for $\alpha$ we'll determine later. By \Cref{thm:mainthm}, with high probability, $\calP$ is $\bigO{\log^2(n) \left(\alpha + n^{\bigO{\alpha^{-1}}}\right)}$-competitive with $R$. Assume it is. Now, by \Cref{cor:integraliseasy}, for every integral demand $d$,
\begin{align*}
    \oblint(\calP, d)   &\leq 2 \oblreal(\calP, d) + 3\ln m\\
                        &\leq \bigO{\log(n)^2 \left(\alpha + n^{\bigO{\alpha^{-1}}}\right)} \obl(R, d) + \bigO{\log n}\\
                        &\leq \bigO{\log(n)^2 \left(\alpha + n^{\bigO{\alpha^{-1}}}\right)} \bigO{\log n} \optint(d)\\
                        &= \bigO{\log(n)^3 \left(\alpha + n^{\bigO{\alpha^{-1}}}\right)} \optint(d).
\end{align*}
Thus, since $n^{\frac{\log \log n}{2 \log n}} = 2^{\frac{1}{2} \log \log n} = \sqrt{\log n}$, there exists $\alpha = \bigO{\frac{\log n}{\log \log n}}$ for which $n^{\bigO{\alpha^{-1}}} = \bigO{\sqrt{\log n}}$. For such $\alpha$, we obtain a $\bigO{\frac{\log n}{\log \log n}}$-sparse integral semi-oblivious routing that is $\bigO{\frac{\log^3 n}{\log \log n}}$-competitive on \zeroonedemands.

\end{proof}

\congdilationcorollary*

\begin{proof}
Assume without loss of generality that $n \geq 4$ and let $s = \left\lceil \frac{\log n}{\log \log n} + 1\right\rceil$. Let $h_1 = 1$ and $h_i = \lceil h_{i - 1} \log n \rceil$ for $i > 1$. For $i \in [s]$, let $R_i$ be a $h_i$-hop oblivious routing with hop stretch $\bigO{\log^7 n}$ and congestion approximation $\bigO{\log^2(n) \log^2(h_i \log n)} = \bigO{\log^4 n}$. By \cite{GhaffariHZ21}, such an oblivious routing exists.

Let $\calP_i$ be the $\alpha$-sparse semi-oblivious routing we obtain by applying \Cref{cor:actualmaincorollary} to $R_i$ and $\alpha = \bigO{\frac{\log n}{\log \log n}}$. Define $\calP$ as $P(s, t) := \bigcup_{i \in [s]} P_i(s, t)$. Now, $\calP$ is $\bigO{\left(\frac{\log n}{\log \log n}\right)^2}$-sparse, and for any demand $d$ and routing $R$, for the minimum integer $j$ such that $h_j \geq \dil(R, d)$ (guaranteed to exist as $h_s \geq n$ and any routing can be made vertex-simple while not increasing congestion or dilation),
\begin{align*}
    \oblreal(\calP_i, d)
        &\leq \bigO{\log(n) \left(\alpha + n^{\bigO{\alpha^{-1}}}\right)} \left(\obl(R_i, d) + 1\right)\\
        &\leq ( \poly \log(n) ) \optint^{(h_i)}(d)\\
        &\leq ( \poly \log(n) ) \obl(R, d).
\end{align*}
Thus, by \Cref{cor:integraliseasy},
\begin{equation*}
    \oblint(\calP_i, d) \leq 2 \oblreal(\calP_i, d) + 3\ln m \leq ( \poly \log(n) ) \obl(R, d).
\end{equation*}
Thus, there exists a routing $R'$ on $\calP_i$ (thus also on $\calP$) that is integral on $d$ such that $\obl(R', d) \leq ( \poly \log(n) ) \obl(R, d)$. But since $R'$ is on $\calP_i$, we have
\begin{equation*}
    \dil(R', d) \leq h_i \poly \log n \leq (\dil(R, d) \log n) \poly \log n.
\end{equation*}
\end{proof}

\congdilationsparsecorollary*

\begin{proof}
Assume without loss of generality that $n \geq 4$ and let $s = \left\lceil \alpha + 1\right\rceil$. Let $h_1 = 1$ and $h_i = \lceil h_{i - 1} n^{\alpha^{-1}} \rceil$ for $i > 1$. For $i \in [s]$, let $R_i$ be a $h_i$-hop oblivious routing with hop stretch $\bigO{\log^7 n}$ and congestion approximation $\bigO{\log^2(n) \log^2(h_i \log n)} = \bigO{\log^4 n}$. By \cite{GhaffariHZ21}, such an oblivious routing exists.

Let $\calP_i$ be the $\alpha$-sparse semi-oblivious routing we obtain by applying \Cref{cor:actualmaincorollary} to $R_i$. Define $\calP$ as $P(s, t) := \bigcup_{i \in [s]} P_i(s, t)$. Now, $\calP$ is $\bigO{\alpha^2}$-sparse, and for any demand $d$ and routing $R$, for the minimum integer $j$ such that $h_j \geq \dil(R, d)$ (guaranteed to exist as $h_s \geq n$ and any routing can be made vertex-simple while not increasing congestion or dilation),
\begin{align*}
    \oblreal(\calP_i, d)
        &\leq \bigO{\log(n) \left(\alpha + n^{\bigO{\alpha^{-1}}}\right)} \left(\obl(R_i, d) + 1\right)\\
        &\leq n^{\bigO{\alpha^{-1}}} \optint^{(h_i)}(d)\\
        &\leq n^{\bigO{\alpha^{-1}}} \obl(R, d).
\end{align*}
Thus, by \Cref{cor:integraliseasy},
\begin{equation*}
    \oblint(\calP_i, d) \leq 2 \oblreal(\calP_i, d) + 3\ln m \leq n^{\bigO{\alpha^{-1}}} \obl(R, d).
\end{equation*}
Thus, there exists a routing $R'$ on $\calP_i$ (thus also on $\calP$) that is integral on $d$ such that $\obl(R', d) \leq n^{\bigO{\alpha^{-1}}} \obl(R, d)$. But since $R'$ is on $\calP_i$, we have
\begin{equation*}
    \dil(R', d) \leq h_i \poly \log n \leq (\dil(R, d) n^{\bigO{\alpha^{-1}}}) \poly \log n = \dil(R, d) n^{\bigO{\alpha^{-1}}}.
\end{equation*}
where the $n^{\bigO{\alpha^{-1}}}$-term absorbs polylogarithmic factors.
\end{proof}

\section{Negative Association}\label{sec:negativeassociation}

A set of random variables being negatively associated is a weaker guarantee than the random variables being independent which still lets us prove properties for the variables similar to those of independent random variables. For an overview of negative association, see \cite{Kum83}.

\begin{definition}[Negatively associated random variables]

Random variables $X = (X_1, \dots, X_n)$ are \textit{negatively associated} if for every two functions $f(X)$ and $g(X)$ that depend on disjoint sets of indices $I$ and $J$ of $[n]$, that are either both monotone increasing or monotone decreasing, we have
\begin{equation*}
    \expec{f(X) g(X)} \leq \expec{f(X)} \expec{g(X)}.
\end{equation*}
\end{definition}

The following two lemmas let us easily show the negative-associatedness of random variables.
\begin{lemma}\label{lem:sumnegativeassociation}
If $X = (X_1, \dots, X_n)$ are zero-one random variables, exactly one of which is $1$, they are negatively associated.
\end{lemma}

\begin{lemma}\label{lem:indnegativeassociation}
If $X = (X_1, \dots, X_n)$ are negatively associated, $Y = (Y_1, \dots, Y_m)$ are negatively associated, and $X$ and $Y$ are independent, then $Z = (X_1, \dots, X_n, Y_1, \dots, Y_m)$ are negatively associated.
\end{lemma}

Applying the definition of negative association to indicator functions on the sum of the relevant indices, we get
\begin{lemma}\label{lem:negativeassociationindep}
Let $X = (X_1, \dots, X_n)$ be negatively associated, $I_1, \dots, I_m$ be disjoint subsets of $[n]$, $X_{I} = \sum_{i \in I} X_i$ and $\gamma_1, \dots, \gamma_m$ be real numbers. Then,
\begin{equation*}
    \prob{\bigcap_{j} X_{I_j} \geq \gamma_j} \leq \prod_{j} \prob{X_{I_j} \geq \gamma_j}.
\end{equation*}
\end{lemma}

Chernoff bounds hold for negatively associated variables just like they hold for independent random variables. We use the two following variants.
\begin{restatable}[Chernoff bound 1]{lemma}{firstchernoffnegativeassociation}\label{lem:firstchernoffnegativeassociation}
Let $X = (X_1, \dots, X_n)$ be negatively associated 0/1-random variables, $I \subseteq [n]$ be a subset of indices, $X_I = \sum_{i \in I} X_i$ and $\mu = \expec{X_I}$. Then, for all $\delta \geq 2$,
\begin{equation*}
    \prob{X_I \geq \delta \mu} \leq \exp\left(-\frac{1}{4} \delta \mu \ln (\delta)\right).
\end{equation*}
\end{restatable}

% Sketch of proof to make sure this really holds
% \begin{proof}
% \begin{equation*}
%    \prob{X_i \geq \delta \mu} \leq \left(\frac{e^{\delta - 1}}{\delta^\delta}\right)^\mu = \exp(-\mu (1 + \delta \ln(\frac{\delta}{e})) = \exp(-\mu (1 + \delta (\ln(\delta) - 1)))
% \end{equation*}
% We thus need to show $1 + \delta (\ln(\delta) - 1) - \frac{1}{4} \delta \ln \delta \geq 0$. That equals $1 + \delta(\frac{3}{4} \ln(\delta) - 1)$. For $\delta = 2$, we get $1 + 2 (\ln(2) - 1) \geq 0$. The derivative w.r.t. $\delta$ is $\frac{1}{4} (3 \ln(\delta) - 1)$. For $\delta \geq 2$, this is $\geq 0$.
% \end{proof}

\begin{restatable}[Chernoff bound 2]{lemma}{secondchernoffnegativeassociation}\label{lem:secondchernoffnegativeassociation}
Let $X = (X_1, \dots, X_n)$ be negatively associated 0/1-random variables, $I \subseteq [n]$ be a subset of indices, $X_I = \sum_{i \in I} X_i$ and $\mu = \expec{X_I}$. Then, for all $\delta > 0$,
\begin{equation*}
    \prob{X_I \geq (1 + \delta) \mu} \leq \exp\left(-\frac{\delta^2 \mu}{2 + \delta}\right).
\end{equation*}
\end{restatable}

\end{document}